\documentclass[]{article}

\usepackage{amsmath, amssymb, bm}
\usepackage{amsthm}
\usepackage{mathtools}
\usepackage{graphicx,color}
\usepackage{multirow}

\def\bSig\mathbf{\Sigma}

\newcommand{\indep}{\perp \!\!\! \perp}

\newtheoremstyle{sf}
{2ex}
{2ex}
{\itshape}
{ }
{\scshape\sffamily\bfseries}
{. }
{ }
{ }

\theoremstyle{sf}

\newtheorem{definition}{Definition}
\newtheorem{remark}{Remark}

\newcommand{\hbeta}{{\widehat\beta}}

\newcommand{\hOmega}{{\widehat\Omega}}

\usepackage{natbib}
\usepackage[colorlinks = true,
            linkcolor = blue,
            urlcolor  = blue,
            citecolor = blue,
            anchorcolor = blue]{hyperref}
\usepackage[margin=1in]{geometry}
\usepackage{setspace}
\usepackage{graphicx}
\usepackage[
singlelinecheck=false,
labelfont=bf,font={doublespacing, small}
]{caption}

\usepackage[affil-it]{authblk} 
\usepackage{etoolbox}
\usepackage{lmodern}

\makeatletter
\patchcmd{\@maketitle}{\LARGE \@title}{\fontsize{16}{19.2}\selectfont\@title}{}{}
\makeatother

\newtheorem{assumption}{Assumption}
\usepackage{sectsty}
\allsectionsfont{\sffamily}

\begin{document}

\title{An Enhanced Cross-Sectional HIV Incidence Estimator that Incorporates Prior HIV Test Results}

\author[1]{Marlena Bannick}
\author[2,3]{Deborah Donnell}
\author[4]{Richard Hayes}
\author[5,6]{Oliver Laeyendecker}
\author[2,3]{Fei Gao}

\affil[1]{Department of Biostatistics, University of Washington, Seattle, Washington, U.S.A.}
\affil[2]{Bioinformatics and Epidemiology Program, Fred Hutchinson Cancer Center, Seattle, Washington, U.S.A.}
\affil[3]{Public Health Sciences Division, Fred Hutchinson Cancer Center, Seattle, Washington, U.S.A.}
\affil[4]{Department of Infectious Disease Epidemiology, London School of Hygiene and Tropical Medicine, London, England, U.K.}
\affil[5]{School of Medicine, Johns Hopkins University, Baltimore, Maryland, U.S.A.}
\affil[6]{Division of Intramural Research, National Institute of Allergy and Infectious Diseases, Baltimore, Maryland, U.S.A.}


\maketitle

\begin{abstract}
    Incidence estimation of HIV infection can be performed using recent infection testing algorithm (RITA) results from a cross-sectional sample. This allows practitioners to understand population trends in the HIV epidemic without having to perform longitudinal follow-up on a cohort of individuals. 
    The utility of the approach is limited by its precision, driven by the (low) sensitivity of the RITA at identifying recent infection. 
    By utilizing results of previous HIV tests that individuals may have taken, we consider an enhanced RITA with increased sensitivity (and specificity). We use it to propose an enhanced estimator for incidence estimation.
    We prove the theoretical properties of the enhanced estimator and illustrate its numerical performance in simulation studies. 
    We apply the estimator to data from a cluster-randomized trial to study the effect of community-level HIV interventions on HIV incidence. We demonstrate that the enhanced estimator provides a more precise estimate of HIV incidence compared to the standard estimator.
\end{abstract}

\section{Introduction}
\label{s:intro}

The past decade has seen tremendous progress in the development of biomedical agents and interventions for HIV prevention.
For example, daily oral Tenofovir-emtricitabine (TDF-FTC) has been shown to have high efficacy for preventing HIV when taken as directed based on results from multiple randomized, placebo-controlled trials
\citep{grant2010preexposure, baeten2012antiretroviral, McCormackSheenaProf2016Pptp, MolinaJean-Michel2015OPPi}.
More recently, cabotegravir as long-acting injectable pre-exposure prophylaxis (PrEP) has been proven highly efficacious in men who have sex with men (MSM) and transgender women (TGW) \citep{landovitz2021cabotegravir} and in women \citep{moretlwe2021long}.
 Determining the real-world effectiveness of these interventions and their community-level implementation strategies, requires monitoring the incidence of HIV over time.
 HIV incidence is typically assessed through longitudinal cohorts, which is resource-intensive and time-consuming.
 Moreover, recruiting individuals to longitudinal cohorts could potentially introduce retention bias or modify their HIV risk, which is also known as the ``Hawthorne effect" \citep{sherr2007voluntary}.
 
 Incidence estimation combining lab-based test results from a cross-sectional sample using a recent infection testing algorithm (RITA) is a promising alternative technique that avoids costly longitudinal follow-up \citep{kassanjee2012new, gao2021statistical}.
 The RITA is comprised of single or multiple biomarker assays that classify HIV infections as recent (for example, infections occurring in the last year) versus non-recent.
 Common RITAs include the limiting-antigen (LAg) avidity assay \citep{duong2015recalibration}, in combination with viral load \citep{kassanjee2016viral,parkin2022facilitating} or other avidity assays \citep{laeyendecker2018identification}. These assays identify recent infections by measuring either the concentration, or the binding strength, of an individual's HIV antibodies, both of which increase with duration of infection \citep{weiDevelopmentTwoAviditybased2010}.
 Cross-sectional incidence estimators \citep{kaplan1999snapshot, kassanjee2012new} can estimate the incidence in the target population, based on a representative cross-sectional sample under mild assumptions regarding the epidemic in the target population and properties of the RITAs \citep{gao2021statistical}.
 
 Accurate incidence estimates using RITAs generally require large sample sizes \citep{parkin2022facilitating} such that the utility of RITAs is often limited by their precision.
 For example, in the HPTN 071 (PopART) study \citep{klockValidationPopulationlevelHIV12021} -- a cluster-randomized trial of the effect of community-level HIV interventions on HIV incidence in Zambia and South Africa -- a cross-sectional sample of 20,472 individuals yielded 73 recent infections based on the RITA that includes the LAg avidity assay and viral load, leading to an incidence estimate of 1.29 per 100 person years with a confidence interval of 0.97-1.62 per 100 person-years.

In addition to the size of the cross-sectional sample, the precision of the cross-sectional incidence estimator also depends on the ability of the RITA to correctly identify recent infections. Previous studies validating several RITAs showed that only 25-46\% of infections occurring in the past year were identified \citep{grant-mcauleyEvaluationMultiassayAlgorithms2021} due to the requirement that a RITA rarely misclassify non-recent infections as recent. Heuristically, this is like balancing sensitivity and specificity. The classification ability of a RITA by could potentially be improved by using available information on timing of HIV infection from prior HIV test results.
Depending on the setting, many individuals in the cross-sectional sample may have had recent prior HIV testing based on records from clinics or other testing centers.
The timing of when a prior HIV test was administered, and its result, are useful pieces of information to strengthen the evidence on whether someone has been recently infected, supplementing the information from the RITA.
As an example, if ``recent'' infection is defined as infection occurring within the last year, and someone had a prior HIV-negative test result six months ago but tests positive today, then they must have been recently infected, even if their recency assay result disagrees. Heuristically, incorporating this additional source of information may help capture more recent infections, leading to more precise estimates.

In this paper, we propose an approach to estimate incidence based on RITAs applied to samples from cross-sectional surveys that also utilizes prior HIV test results.
We first define a new algorithm that classifies recent infection using both recency assay results and available prior HIV test results.
This new algorithm has higher sensitivity and specificity than those of the original RITA.
Then, we construct an incidence estimator based on this new algorithm, by following a similar derivation as in \citet{gao2021statistical}.
We articulate the assumptions required for our enhanced estimator to be valid and examine the performance of the enhanced estimator under realistic simulation scenarios, including under violation of such assumptions.
Finally, we apply our enhanced estimator and the standard estimator to data from the HPTN 071 (PopART) study \citep{hayesEffectUniversalTesting2019, grant-mcauleyEvaluationMultiassayAlgorithms2021}.

\section{Methodology}
We propose an enhanced cross-sectional incidence estimator that builds on the framework used to derive the commonly used \citet{kassanjee2012new} estimator, which we will refer to as the ``standard estimator.'' First, we introduce the methodological framework for the standard estimator. In Section \ref{sec:kassanjee}, we briefly review the intuition behind the derivation of the standard estimator. In Section \ref{sec:enhanced}, we propose a new estimator that further incorporates prior HIV testing results.

Consider a cross-sectional sample that includes $N$ independent individuals. Without loss of generality, assume that the cross-sectional sample is taken at one point in calendar time for all individuals, denoted as $t_{cs}$.
Each individual is first tested for HIV infection.
Let $D_i$ be the indicator of HIV status for individual $i=1,\dots,N$, such that $N_{pos} = \sum_{i=1}^N D_i$ individuals are HIV-positive and $N_{neg} \equiv N - N_{pos}$ individuals are HIV-negative. Let $I_i$ be the calendar infection time for an HIV-positive individual $i$, with $I_i \leq t_{cs}$.
Let $U_i = t_{cs} - I_i$ be the (unobserved) duration of infection in an HIV-positive individual $i$ at time $t_{cs}$. For an individual who is found to be HIV-positive, i.e., $D_i=1$, they are tested using a RITA to determine whether this infection is ``recent'', i.e., with a duration of at most $T^*$.
Here, $T^*$ is a pre-specified cut-off time for indicating a relevant time span for incidence estimation which is usually set to be 2 years. 
\begin{definition}[Recently Infected]\label{def:recent-infection}
    We refer to an HIV-positive individual as ``recently infected'' if their duration of infection is no more than $T^*$, i.e., $U_i \leq T^*$. We refer to all other HIV-positive individuals as non-recently infected.
\end{definition}
Let $R_i$ be the indicator of being identified as a recent infection by the RITA for individual $i$ with $D_i = 1$, and $N_{rec} = \sum_{i:D_i=1} R_i$ be the total number of RITA-recent individuals (Definition \ref{def:rita-recent}). Thus, the full data from the cross-sectional sample at calendar time $t_{cs}$ is $(D_i, R_i)$, $i = 1, ..., n$, where $R_i$ is only defined for individuals with $D_i = 1$.

\begin{definition}[RITA-recent]\label{def:rita-recent}
    We refer to an HIV-positive individual $(D_i = 1)$ as RITA-recent if they are identified as recent by the recent infection testing algorithm (RITA), i.e., $R_i = 1$.
\end{definition}

\subsection{\citet{kassanjee2012new} Estimator}\label{sec:kassanjee}

Two important properties of a RITA are its mean duration of recent infection (MDRI), and false recent rate (FRR). MDRI is the mean duration of RITA-recency amongst those recently infected individuals, calculated as $$\Omega_{T^*} = \int_0^{T^*}\Pr(R_i = 1 | U_i = u)du.$$ We use the notation of $\phi(u) := \Pr(R_i = 1 | U_i = u)$ defined as the probability of testing RITA-recent for a random individual infected $u$ time units ago, so that $\Omega_{T^*} = \int_0^{T^*} \phi(u) du$. FRR is the probability of testing RITA-recent for a randomly selected, non-recently infected individual, and is often assumed to be a constant value, $\beta_{T^*}$. Estimates of the MDRI and FRR for a RITA are denoted as $\hOmega_{T^*}$ and $\hbeta_{T^*}$, respectively, which are commonly constructed based on external studies of individuals that have known infection duration.

The probability of being identified as RITA-recent $\Pr(R_i=1)$ can be expressed as a function of incidence, prevalence, MDRI, and FRR \citep{kassanjee2012new, gao2021statistical}. 
A cross-sectional incidence estimator can be derived by solving for incidence in that expression and plugging in data from a cross-sectional sample $\{N_{neg}, N_{pos}, N_{rec}\}$, and the external estimates $\hOmega_{T^*}$ and $\hbeta_{T^*}$ (equation 25 in \citet{kassanjee2012new}):
\begin{equation}
 \tilde\lambda = \frac{N_{rec} - N_{pos}\hbeta_{T^*}}{N_{neg}(\hOmega_{T^*} - \hbeta_{T^*}T^*)}.\label{equ:kassanjee}
\end{equation}
The standard estimator above \citep{kassanjee2012new} is consistent for the incidence among the target population given a ``steady state'' condition of the epidemic \citep{gao2021statistical}. The precision of this estimator is related to the number of correctly identified RITA-recent individuals, which increases with increasing MDRI and decreasing FRR \citep{gao2020sample}.

\subsection{Incorporating prior HIV testing results}\label{sec:enhanced}

We propose an enhanced cross-sectional incidence estimator that uses the same data as the standard estimator, as well as information from prior HIV tests for HIV-positive individuals. We first introduce a new algorithm for identifying recent infection using information from both a RITA and prior HIV tests, which we call the PT-RITA (Definition \ref{def:PT-RITA-recent}). We discuss the characteristics of the PT-RITA, relative to the characteristics of the original RITA. Finally, we construct a new cross-sectional incidence estimator based on the PT-RITA, using the same intuition used to construct the original standard estimator $\tilde\lambda$.

\subsubsection{PT-RITA: An enhanced algorithm for recency testing}\label{sec:algorithm}
Let $Q_i$ be the indicator whether HIV-positive individual $i$ has a prior HIV testing result. In this paper, we consider only one prior HIV testing result per HIV-positive individual.
For individuals with $Q_i=1$, let $T_i$ be the time between their prior HIV test and $t_{cs}$, the time of the cross-sectional survey. We assume a bounded support for $T_i \in [0, \tau]$, $0 < \tau < \infty$. In other words, HIV-positive individuals report prior HIV test results from no more than $\tau$ time units ago. Let $\Delta_i=1$ and $\Delta_i = 0$ indicate an HIV-positive and HIV-negative result at $T_i$, respectively.
We emphasize that $Q_i$, $T_i$, and $\Delta_i$ are each random variables with realizations from i.i.d. observations across individuals. Note that just as with a RITA result, prior testing only provides additional information when individual $i$ is HIV-positive (i.e., $D_i=1$). HIV-negative individuals are not tested with the RITA, and likewise we would not use any prior HIV tests that they report. Thus, among an HIV-positive individual, their prior testing data can be represented as the triple $(Q_i, Q_iT_i,Q_i\Delta_i)$.

Recall that $R_i = 1$ indicates RITA-recency. We define two new indicators among those with $D_i = 1$ to help construct the new algorithm:
\begin{itemize}
    \item Let $R^*_i = I(R_i=0, T_i\leq T^*, \Delta_i =0, {Q_i = 1}$). These are HIV-positive individuals who were not RITA-recent, but who we know are recently infected because they tested negative for HIV within the recent $T^*$ time period (Definition \ref{def:recent-infection}). These individuals should be \textit{added} to the recent infection group.
    \item Let $L_i = I(R_i=1, T_i \geq T^*, \Delta_i = 1, {Q_i = 1})$. These are HIV-positive individuals who are RITA-recent, but who are known to be non-recently infected because they tested positive for HIV at least $T^*$ time units ago. These individuals should be \textit{removed} from the recent infection group.
\end{itemize}

Based on $R_i$, $R_i^{*}$, and $L_i$, we construct an enhanced recency testing algorithm that incorporates prior test results from HIV-positive individuals, defined in Definition \ref{def:PT-RITA-recent}.

\begin{definition}[PT-RITA-recent]\label{def:PT-RITA-recent}
    We refer to an HIV-positive individual $(D_i = 1)$ as PT-RITA-recent if $R_i^{PT} = 1$, where
    \begin{align*}
        R_i^{PT} = (R_i - L_i) + R_i^*.
    \end{align*}
    In other words, $R_i^{PT} = 1$ means that individual $i$ either (1) tested RITA-recent and is not determined to be non-recently infected based on their HIV testing history, or (2) is known to be recently infected based on their HIV testing history. Individuals with $L_i = 1$ are removed from the RITA-recent set, and individuals with $R_i^*$ are newly added to the RITA-recent set.
\end{definition}

Let $N^{PT}_{rec} = \sum_{i:D_i=1} R_i^{PT}$ be the total number of PT-RITA-recent individuals. Using PT-RITA-recency to classify recent infections can be viewed as a ``better'' algorithm than RITA-recency: it is better at distinguishing between recent and non-recent infections in terms of MDRI and FRR (shown in Section \ref{sec:ptrita-characteristics}).

\subsubsection{Characteristics of the PT-RITA}\label{sec:ptrita-characteristics}

The MDRI of the PT-RITA is defined as $\Omega_{T^*}^{PT} = \int_0^{T^*}\Pr(R_i^{PT} = 1|U_i = u) du$.
Our derivations show that the MDRI of this new recency testing algorithm can be written as (see Web Appendix A)
\begin{align}
    \begin{split}
    \Omega_{T^*}^{PT} 
      &= \Omega_{T^*} + E \left\{Q_i I(T_i \leq T^*) \int_0^{T_i}\{1-\phi(u)\} du\right\} 
      ,\label{eq:mdri-2}
      \end{split}
\end{align}
where the expectation is taken over the random variables $Q_i$ and $T_i$, and where $\Omega_{T^*}$ and $\phi(u)$ are the MDRI and test-recent functions from the original RITA.
From this expression, we see that if HIV-positive individuals have available prior HIV-negative tests within the $T^*$ recent past, the MDRI of the PT-RITA is larger than the MDRI of the RITA.
If $T_i = T^*$ for all $i$ with $Q_i=1$, then $\Omega_{T^*}^{PT} = T^*$, i.e., the new algorithm has the highest possible MDRI (perfect sensitivity) regardless of the function $\phi(u)$. This makes intuitive sense: if everyone is tested for HIV at $T^*$, then we definitively know who is infected within the last $[0, T^*]$ time units, and do not need to rely on the recency test result $R_i$. We also show that the FRR of the PT-RITA is no larger than the FRR of the RITA (see Web Appendix A).

\subsubsection{Incidence estimation based on the PT-RITA}

We now construct an incidence estimator based on the PT-RITA, which uses a similar derivation to the standard estimator \citep{kassanjee2012new, gao2021statistical}.
To derive the enhanced incidence estimator, we rely on the following two assumptions, similar to the assumptions necessary for the \citet{kassanjee2012new} estimator.

\begin{assumption}[Constant FRR]\label{as:constant-frr}
The probability of testing RITA-recent for an HIV-positive, non-recently infected individual is constant with respect to infection duration, i.e., $u \geq T^* \implies \phi(u) := \beta_{T^*}$.
\end{assumption}

\begin{assumption}[Uniformly Distributed Infection Times]\label{as:constant-inc}
For $u \in [0, \max(\tau, T^*)]$, $\Pr(I_i = t_{cs} - u | I_i \leq t_{cs}) = \Pr(I_i = t_{cs} | I_i \leq t_{cs})$, where $\tau$ is the upper bound of the support for the random variable $T_i$.
\end{assumption}

Assumption \ref{as:constant-frr} and a variation of Assumption \ref{as:constant-inc} with $u \in [0,T^*]$ are required for the \citet{kassanjee2012new} estimator. Assumption \ref{as:constant-inc} is implied if incidence and prevalence are constant within the time period of interest \citep{gao2021statistical}.
In addition, we make the following assumptions on the prior testing distribution and how it relates to the recency assays.
\begin{assumption}[Conditional Independence of Recency Assay Results and Prior Test Availability and Timing]\label{as:cond-indep}
The RITA result for an individual who is HIV-positive at $t_{cs}$ is independent of whether they had a prior HIV test and how long ago that test was taken, conditional on their infection duration, i.e., $R_i \indep (T_i, Q_i) | U_i$.
\end{assumption}

\begin{assumption}[Independence of Infection Timing and Prior Testing Availability and Timing]\label{as:independence-testing}
The availability and time of prior testing for an HIV-positive individual at $t_{cs}$, $({Q_i}, T_i)$, is independent of their infection duration $U_i$, i.e., $({Q_i},T_i) \indep U_i$.
\end{assumption}
As Assumption \ref{as:cond-indep} and \ref{as:independence-testing} are related to prior testing availability and timing, they only apply to those that are HIV-positive at time $t_{cs}$. Assumption \ref{as:cond-indep} requires independence between the RITA results and having had a recent prior HIV test (conditional on infection duration).
This assumption is likely to hold in practice, as it is reasonable to assume that for an individual who is infected, their biological reaction to the RITA is not influenced by whether they have had a prior HIV test or how long ago they were tested.
Assumption \ref{as:independence-testing} requires that whether and when the individual had a prior HIV test is independent of when they were infected. It is conceivable that Assumption \ref{as:independence-testing} will be violated in practice: non-recently infected individuals are less likely to have a recent HIV test. For example, individuals may stop taking HIV tests once they receive a positive result; individuals who have a higher HIV risk are possibly more likely to seek testing after HIV exposure.
We will evaluate impact of violation of Assumption \ref{as:independence-testing} in simulation studies with a data generating mechanism that induces this dependent relationship between prior test timing and infection duration. We later discuss how we might relax this assumption in future work.

Recall that the standard estimator can be derived by expressing the probability of testing RITA-recent $\Pr(R_i=1)$ in terms of epidemiological parameters and characteristics of the RITA. We perform a similar derivation: we express the probability of testing PT-RITA-recent $\Pr(R_i^{PT} = 1)$ in terms of the same epidemiological parameters, but using characteristics of the PT-RITA instead (see Web Appendix B).
We obtain an estimator $\hat\lambda$ which we refer to as the enhanced estimator:
\begin{align}\label{est}
    \hat\lambda = \frac{N_{rec}^{PT} - \sum_{i:D_i=1}\hat{\beta}_{T^*}\left\{1-I(T_i \geq T^*)Q_i\right\}}{N_{neg}\left(\hat{\Omega}_{T^*}  + \sum_{i:D_i=1}\left[Q_i \cdot I(T_i\le T^*)\int_{0}^{T_i} \{1-\hat{\phi}(u)\} du - \hat{\beta}_{T^*}\{T^* - Q_i \cdot I(T_i\geq T^*)T_i\}\right]/N_{pos}\right)}.
\end{align}
The expression for $\hat{\lambda}$ contains characteristics of the original RITA: $\hat{\Omega}_{T^*}$, $\hat{\beta}_{T^*}$, and $\hat{\phi}(u)$. This makes intuitive sense: the PT-RITA combines the data from the original RITA with prior HIV testing data. We derive an analytical form for the variance of $\hat{\lambda}$ in Web Appendix C, which can be estimated with sample statistics. 

Notice that use of $\hat{\lambda}$ requires an estimate of the entire $\phi(t)$ function for the original RITA, rather than just its integral $\Omega_{T^*}$. As with the standard estimator, we require consistency of $\hat{\Omega}_{T^*} \to \Omega_{T^*}$ and $\hat{\beta}_{T^*} \to \beta_{T^*}$. In addition, in order for $\hat{\lambda}$ to be consistent for $\lambda(s)$, the incidence at the cross-sectional sample time, it suffices for the following mild condition (Assumption \ref{as:consistency}) on $\hat{\phi}$ and $\phi$ to hold (which is satisfied for $\phi$ estimated with a generalized linear model). See Web Appendix B.2 for a proof of consistency of $\hat{\lambda}$.
\begin{assumption}[Consistency of $\hat{\phi}$]\label{as:consistency}
    We assume that $\hat{\phi}$ and $\phi$ are Lebesgue integrable functions on $[0, \tau]$, and that $\int_0^{\tau} |\hat{\phi}(u) - \phi(u)| du \to 0$ in probability with respect to the randomness in the data used to estimate $\hat{\phi}$.
\end{assumption}

\begin{remark}
    If no one has a prior HIV test available ($Q_i = 0$ for $i:D_i=1$), then $\hat{\lambda}$ reduces to the standard estimator $\tilde{\lambda}$ in \eqref{equ:kassanjee}.
    Furthermore, if $Q_i = 1$ and $T_i = T^*$ for $i:D_i=1$, (i.e., everyone took a prior test at exactly time $T^*$), then
    $$\hat{\lambda} = \frac{N^{PT}_{rec}}{ N_{neg} T^*},$$
    which is the standard estimator using a recency assay with a ``perfect'' MDRI equal to $T^*$, and an FRR of zero. Indeed, the data are similar to that collected in a longitudinal cohort where each individual's HIV status is known at the cohort entry and end of follow-up. However, the estimator above is slightly different from a simple incidence estimator (cases per person-years): the denominator involves HIV negative individuals at the end of follow-up rather than at cohort entry. The difference between these approaches is discussed in \citet{gao2020sample} as corresponding to an open versus closed cohort.
\end{remark}

\begin{remark}
    \citet{gao2021statistical} derived the ``mean shadow time'', $\omega^*$, of the standard estimator. Its interpretation is as follows: under linearly decreasing incidence over time, the incidence estimator at time $t_{cs}$ is consistent for incidence at time point $t_{cs} - \omega^*$. We use a similar derivation (see Web Appendix B.3) to show that the mean shadow time for the enhanced estimator is
\begin{align*}
    \omega^* := \frac{\left[\int_0^{T^*} \left\{\phi(u) - \beta_{T^*} \right\}u du \right] + \left[\int_{0}^{\tau} u E\left\{Q_i 1(u \leq T_i) \left[1(T_i \leq T^*) \left(1-\phi(u) \right) + 1(T_i \geq T^*) \beta_{T^*} \right] \right\} du \right]}{\left[\Omega_{T^*} -\beta_{T^*} T^*\right] + \left[E \left\{Q_i I(T_i \leq T^*) \int_{0}^{T_i} \{1-\phi(u)\} du\right\} + \beta_{T^*}E\left\{Q_i I(T_i \geq T^*) T_i\right\}\right]}.
\end{align*}
    The first terms in the numerator and denominator are equivalent to the mean shadow time of the adjusted estimator. The second terms in the numerator and denominator are functions of the distribution of prior HIV testing times and availability.
\end{remark}

\section{Simulation Study}

We perform a simulation study to assess the performance of the enhanced estimator $\hat{\lambda}$ in \eqref{est} and the standard estimator from \cite{kassanjee2012new}, in a variety of settings. We have two primary aims of these simulations: (1) to establish under what circumstances there are efficiency gains to using $\hat{\lambda}$ over $\tilde{\lambda}$, and (2) to assess the impact of violating key assumptions required for using $\hat{\lambda}$.

We generate cross-sectional data on HIV status at time $t_{cs}$ and infection timing according to Section 3.3 of \citet{gao2021statistical}, using the constant incidence described in their Section 3.1. Briefly, incidence is set to 0.032 cases per person-year, and prevalence is set to 0.29. We generate RITA data for HIV-positive individuals at time $t_{cs}$ according to the RITA characteristics 1B from \citet{gao2021statistical} Section 3.2. The RITA has an MDRI of 98 days and an FRR of 1.4\%. Based on these RITA characteristics, we use a similar external data simulation procedure to get $\hat{\phi}(t)$ described in Section 3.3.1 of \citet{gao2021statistical}. Instead of simply reducing $\hat{\phi}(t)$ to an estimate of MDRI (as was done in \citet{gao2021statistical}), we use the entire estimated $\phi(t)$ function (and its variance) in the point estimate and variance estimate for the enhanced estimator.

To generate data on prior HIV tests for all individuals who are HIV positive based on the cross-sectional data, we have: the test availability indicator $Q_i\sim \mathrm{Bernoulli}(q)$ where $q$ is the probability of having a prior HIV test available; the time of the prior HIV test $T_i\sim \mathrm{Unif}[a,b]$ where $a$ and $b$ are the min and max prior testing times; and the prior test result among those with $Q_i = 1$ as $\Delta_i = 1(T_i \leq U_i)$, i.e., $\Delta_i = 1$ if their prior test was taken after their infection time $I_i$.

\subsection{Efficiency Gains of the Enhanced Estimator}

Here we explore availability of prior test results for HIV-positive individuals and how they affect the efficiency of the enhanced estimator relative to the standard estimator.
Specifically, we consider:
\begin{enumerate}
	\item Varying probabilities $q$ of having prior test results available. We consider $q$ of 0.2, 0.4, 0.6, 0.8, and 1.0 to illustrate how the efficiency gain of the enhanced estimator increases with availability of prior testing.
	\item Varying intervals $[a,b]$ during which test results are available. We consider three scenarios: $[0, 2]$, $[0, 4]$, and $[2, 4]$ to illustrate how the efficiency gain of the enhanced estimator depends on prior testing timing. Specifically, $[0, 2]$ represents a case where all prior tests occur within the window of time that defines a recent infection, $[0, T^*]$.
\end{enumerate}
The results of these simulations are presented in Figure \ref{fig:main-sim}, with further details that show the performance of the variance estimator in Web Appendix D Table A1.
The incidence estimates of the standard estimator range from -0.01 to 0.08 cases per person-year (negative values may occur when there are very few individuals identified as RITA-recent). The bias of the standard estimator is 0.001 cases per person-year. The mean squared error (MSE) of the standard estimator is 0.0001.
The enhanced estimator has a similar amount of bias (which reduces with increasing proportion of prior tests available). Importantly, the enhanced estimator has a substantial reduction in MSE compared with the standard estimator, even when a small fraction of HIV-positive individuals has prior test results available. For example, using the enhanced estimator leads to an MSE reduction of 25\% when 20\% of HIV positive people have test results within the last four years. If 80\% of HIV individuals had a test within the last two years, the MSE reduction could be as large as 63\%. Based on Figure \ref{fig:main-sim}, we see that it is more efficient to use prior HIV tests within the ``recent'' time period than non-recent time period. For example, for a recent infection defined as within the last two years, it is more advantageous if 20\% of study participants have prior test results within the recent infection window, i.e., 0-2 years (58\% reduction in MSE) than within 2-4 years (19\% reduction in MSE). These large gains in precision do not come with any added bias, if the assumptions hold. 

\begin{figure}[htbp!]
\caption{Simulation results (5,000 simulated datasets) comparing the median point estimates, interquartile range, and reduction in MSE by using the enhanced estimators across range of prior test times (uniformly distributed within the range $(a, b)$) available for $q$ fraction of participants.}\label{fig:main-sim}
\includegraphics[width=\textwidth]{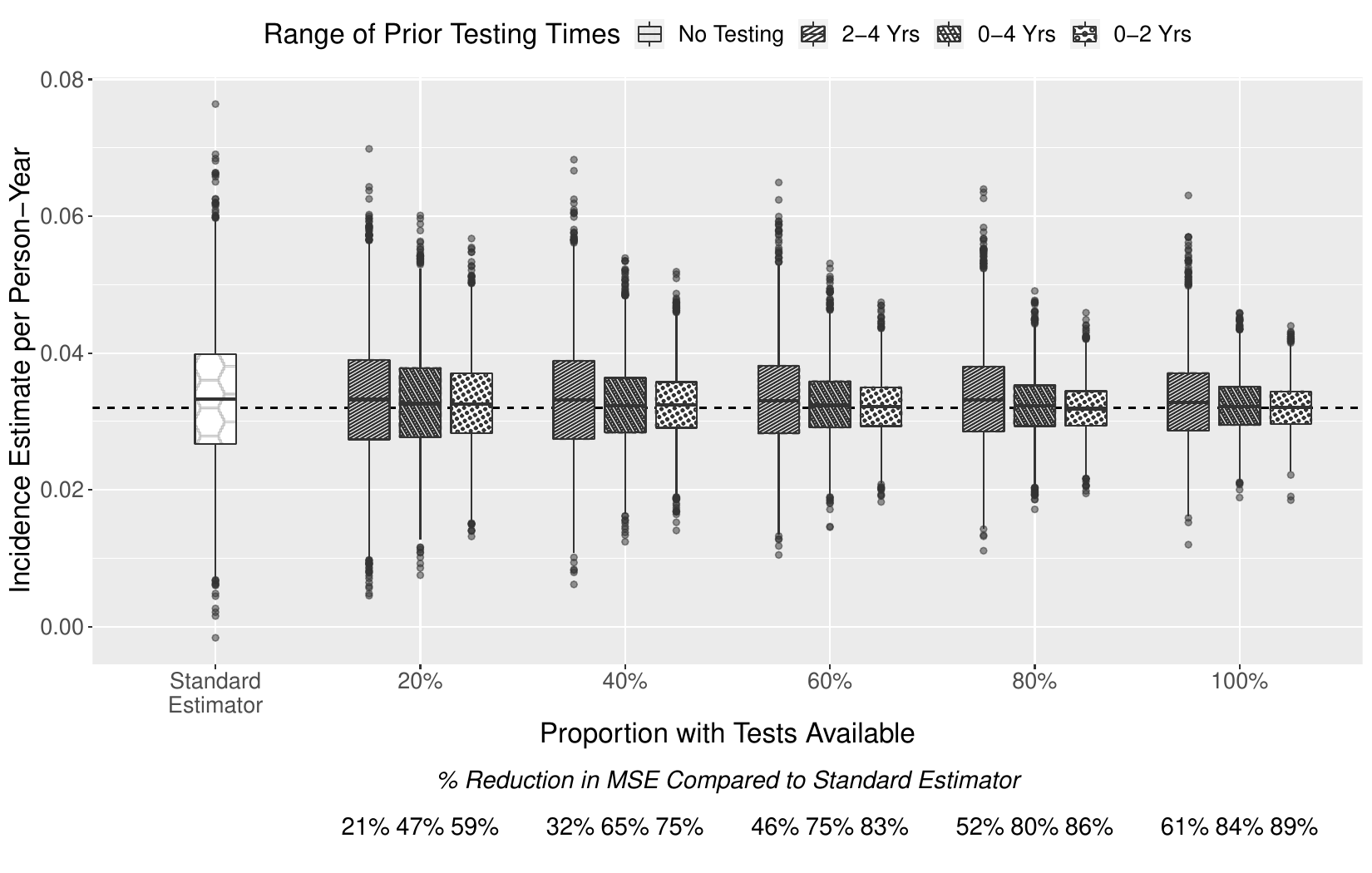}
\end{figure}

\subsection{Sensitivity of the Enhanced Estimator}

The enhanced estimator requires several assumptions beyond those which are required by the standard estimator and which are outlined in \citet{gao2021statistical}. We list these assumptions in Table \ref{tab:assumptions}. Most importantly, the enhanced estimator requires assumptions on the independence of characteristics of prior test results and infection characteristics.
We perform a sensitivity analysis with violations of Assumption \ref{as:constant-inc} (Section \ref{sens:constant-inc}) and Assumption \ref{as:independence-testing} (Section \ref{sens:independence-testing}). We do not explore violation of Assumption \ref{as:cond-indep} because it is very likely to hold in practice.

An implicit assumption in using the enhanced estimator is that the prior HIV testing data are accurate. We denote this in Table \ref{tab:assumptions} as Assumption $\star$. If data are collected only from lab measurements (e.g., study-administered HIV tests, records from HIV clinics), then this is not a concern. However, if data are self-reported by study participants, then there could be inaccuracies in the prior testing data.
In practice, the validity of testing results may be compromised due to recall bias, HIV stigma, or a desire to be eligible for a trial, especially when prior HIV tests are self-reported.
To assess the impact of this, we perform a sensitivity analysis where the prior HIV testing data are subject to recall bias in several dimensions.

\begin{table}
\centering
\caption{Comparison of assumptions required by standard and enhanced estimators}\label{tab:assumptions}
\begin{tabular}{ll|cc}
\hline
\multirow{2}{*}{Assumption} & \multirow{2}{*}{Description} & Standard & Enhanced \\
& & Estimator & Estimator \\
\hline
\hline
	\textbf{Assumption \ref{as:constant-frr}} & Constant FRR & $\checkmark$ & $\checkmark$ \\
	\textbf{Assumption \ref{as:constant-inc}} & Uniformly distributed infection times & Until $T^*$ & Until $\max(\tau, T^*)$\\
	\multirow{2}{*}{\textbf{Assumption \ref{as:cond-indep}}} & (Recency assay $\perp$ prior test & \multirow{2}{*}{$\times$} & \multirow{2}{*}{$\checkmark$} \\
    & timing and availability) $|$ infection duration & & \\
	\multirow{2}{*}{\textbf{Assumption \ref{as:independence-testing}}} & Infection timing $\perp$ prior test & \multirow{2}{*}{$\times$} & \multirow{2}{*}{$\checkmark$} \\
    & timing and availability & & \\
    \textbf{Assumption $\star$} & Validity of prior test results & $\times$ & $\checkmark$ \\
\hline
\end{tabular}
\end{table}

\subsubsection{Non-Constant Incidence (Assumption \ref{as:constant-inc})}\label{sens:constant-inc}

Assumption \ref{as:constant-inc} on uniformly distributed infection times is violated if incidence is non-constant \citep{gao2021statistical}. We focus on a scenario where incidence is constant from time $t_{cs} - T^*$ until $t_{cs}$, with $T^* = 2$ years, and linearly decreasing prior to time $t_{cs} - T^*$ (see Web Appendix E, Figure A5). In this setting, the constant incidence assumption of the standard estimator holds. However, if the upper bound of the support for $T_i$ ($\tau$) is greater than $T^*$, then the constant incidence assumption for the enhanced estimator is violated. We provide details on generating infection times according to this new incidence function in Web Appendix E.

Our results in Web Appendix D Table A2 show that the bias of the enhanced estimator is relatively small, except in the scenarios where the great majority of HIV-positive individuals have prior test results from several years ago (when incidence was much different than the current time). However, we also considered using the enhanced estimator and ignoring prior results where $T_i > T^*$ (``only recent''). This represents a practical scenario where the investigator knows that incidence has not been constant all the way back to the oldest prior test, but suspects that it has been relatively constant in recent years. Indeed, we can reduce the bias by excluding non-recent test results, with some small increase in variance.

\subsubsection{Dependence between Prior Testing and Infection Timing (Assumption \ref{as:independence-testing})}\label{sens:independence-testing}

Assumption \ref{as:independence-testing} states that if and when an HIV-positive individual took a prior HIV test is independent of when they became infected.
We examine the scenario when Assumption \ref{as:independence-testing} is violated by simulating HIV testing time from a mixture of two different testing mechanisms, with results in Table \ref{tab:scenario-a} (see Web Appendix E for details).
On top of the background testing where testing times are generated from a uniform distribution (``Base'' mechanism in Table \ref{tab:scenario-a}), there is another testing mechanism where individuals are more likely to take a test after they are infected with HIV, rather than before (``RI'' mechanism in Table \ref{tab:scenario-a}). This second infection-time-dependent testing mechanism could arise if individuals seek out testing after experiencing HIV-related symptoms. The combination of these two mechanisms results in negative bias: individuals are systematically more likely to have a positive prior HIV test than they are to have a negative prior HIV test. As a result, $L_i$ (known non-recent infection based on prior HIV tests) is indicated more frequently than $R_i^*$ (known recent infection based on prior HIV tests). Therefore, $N_{rec}^{PT}$ will be artificially smaller than $N_{rec}$, leading to an underestimate of incidence.

\begin{table}[htbp!]
\centering
\caption{Simulation results (5000 simulated datasets): comparing bias, standard error, and MSE of the standard and enhanced estimators across prior testing mechanisms. Under the ``Base'' mechanism, all HIV-positive individuals at time $t_{cs}$ received HIV tests within the past $[0, 4]$ years with probability $q$. Under the ``Base + RI'' mechanism, they are additionally given a test based on their infection time (see details in Web Appendix E). Given that there are two testing mechanisms in ``Base + RI'', the probability of a test being available is not the same as $q$; thus, $q^*$ is the average proportion of tests available across the simulations}\label{tab:scenario-a}
\begin{tabular}{cc|c|cccc}
  \hline
  $q$ & Mechanism & $q^*$ & Bias$\times 10^2$ & SE$\times 10^2$ & MSE$\times 10^4$ & Coverage \\ 
  \hline
  \multicolumn{7}{c}{\textit{Standard Estimator}} \\
  \hline 
-- & -- & -- & 0.138 & 1.006 & 1.030 & 95.140 \\ 
  \hline
  \multicolumn{7}{c}{\textit{Enhanced Estimator}} \\
   \hline 
0.25 & Base & 0.250 & 0.091 & 0.710 & 0.512 & 94.960 \\ 
  0.25 & Base + RI & 0.633 & -0.759 & 0.493 & 0.820 & 63.800 \\ 
   \hline 
0.5 & Base & 0.500 & 0.045 & 0.552 & 0.307 & 95.360 \\ 
  0.5 & Base + RI & 0.755 & -0.573 & 0.448 & 0.529 & 72.420 \\ 
   \hline 
0.75 & Base & 0.750 & 0.026 & 0.473 & 0.225 & 94.640 \\ 
  0.75 & Base + RI & 0.878 & -0.424 & 0.416 & 0.353 & 80.120 \\
  \hline\end{tabular}
\end{table}

\subsubsection{Measurement Error in Prior HIV Testing Data (Assumption $\star$)}


We introduce the notation $(Q_i', Q_i'T_i', Q_i'\Delta_i')$ to denote prior testing data that are self-reported. We assume that the time of prior test is reported with an error of a standard deviation $\gamma$ of 0, 1 or 6 months. That is, we observe $T'_i$ in place of $T_i$, where {$T'_i = \min(0, T_i + \epsilon_i)$, with $\epsilon_i \sim N(0, \gamma)$}. We consider cases where 10\% of HIV-positive individuals with positive prior HIV tests incorrectly report never having been tested ($Q_i' = \xi_i \cdot Q_i \Delta_i $ with $\xi_i \sim \text{Binom}(0.9)$), or where 10\% of HIV-positive individuals with prior positive prior HIV tests incorrectly report a prior \textit{negative} HIV test result ($\Delta_i' = \eta_i \cdot Q_i \Delta_i $ with $\eta_i \sim \text{Binom}(0.9)$).
This last scenario is the most extreme, since individuals selectively change the result of their test. And as a result, in Figure \ref{fig:scenario-b}, we see that this scenario leads to the largest increase in bias. The addition of random error to the prior test time does not lead to a substantial increase in bias, but it does slightly decrease the precision of the enhanced estimator.

\begin{figure}[htbp!]
\caption{Simulation results (5,000 simulated datasets) comparing bias, interquartile range, and coverage by using the enhanced estimator under recall bias in $(Q_i', Q_i'T_i', Q_i'\Delta_i')$. In these scenarios, the prior test range is uniformly distributed between 0 and 4 years, and 50\% of HIV-positive individuals truly have a prior test available.}
\includegraphics[width=\textwidth]{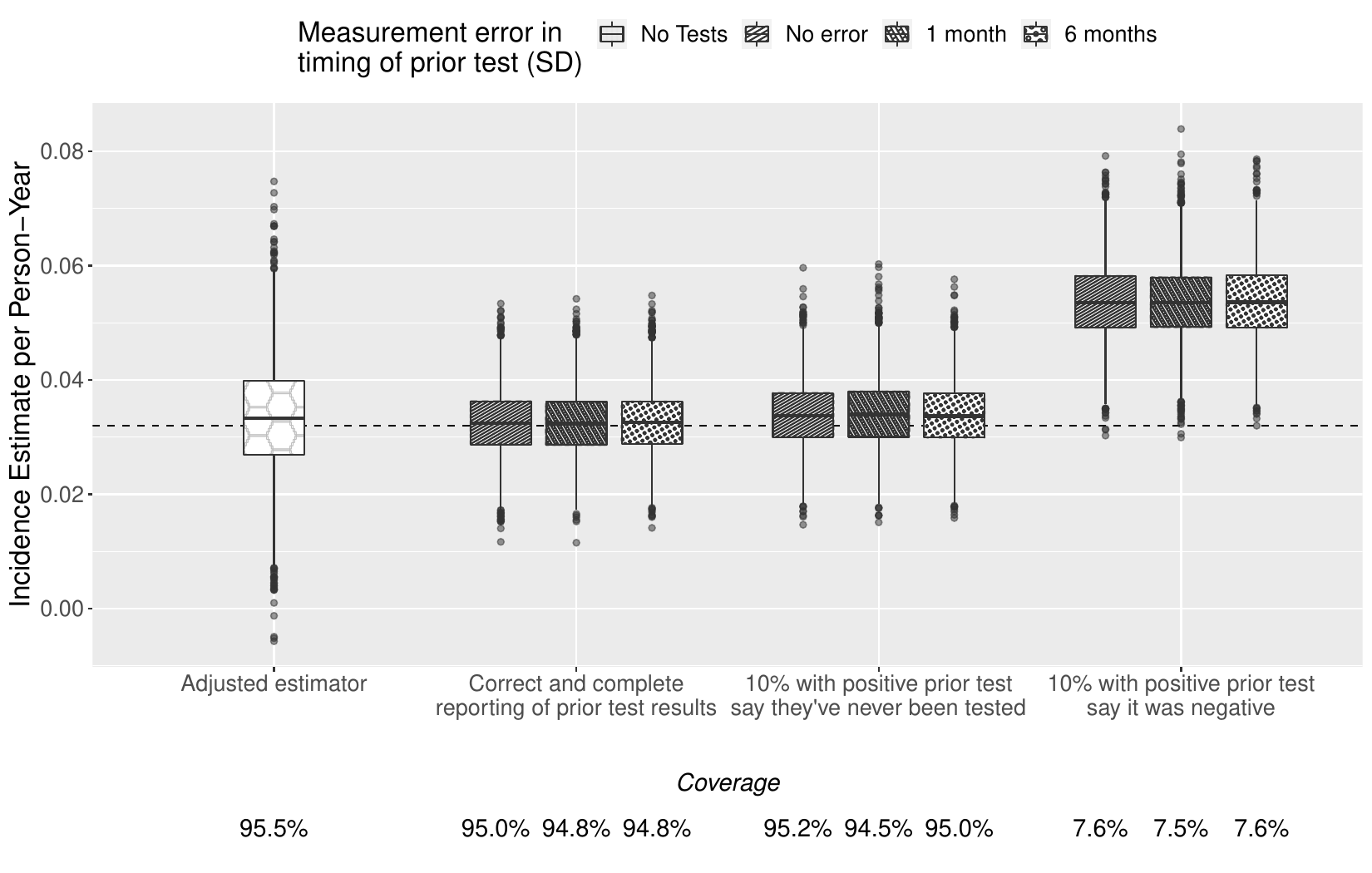}\label{fig:scenario-b}
\end{figure}

\section{Illustration with Trial Data}

\subsection{HPTN 071 PopART Data for HIV and RITA Data}

We apply the enhanced estimator to data from the HIV Prevention Trials Network Study 071, Population Effects of Antiretroviral Therapy to Reduce HIV Transmission (PopART). PopART was a cluster-randomized trial to study the effect of community-level HIV interventions on HIV incidence in Zambia and South Africa \citep{hayesHPTN071PopART2014}. Interventions randomized at the community level were implemented over the course of three years. HIV incidence was assessed via a longitudinal population cohort. HIV-positive participants also had HIV recency assessed at baseline (PC0) and 24 months (PC24) using a RITA. Previous research has validated the use of the standard estimator for cross-sectional HIV incidence estimation using PopART data \citep{klockValidationPopulationlevelHIV12021}.

PopART data also provide a unique opportunity to validate and assess the real-world performance of the enhanced estimator. Figure \ref{fig:hptn-scheme} shows the PopART data availability. At two years, participants within the longitudinal cohort had been tested twice for HIV: once at baseline (PC0), and once at follow-up year one (PC12). We use the PC12 study-administered HIV test calendar time and results as prior testing data, along with the HIV test and RITA result at PC24, in our enhanced estimator for estimating incidence at the 24-month survey. We compare this with the standard estimator that uses only the cross-sectional HIV test and RITA result at PC24, and a longitudinal incidence estimator based on number of events over person-time at risk (which acts as the validation incidence) between PC12 and PC24. One difference between the longitudinal and cross-sectional estimators is who is included in the analysis. The longitudinal estimator requires that individuals were present at PC12 and PC24, whereas the cross-sectional estimators only require that individuals were present at PC24. We also consider a scenario where we randomly select only half of the prior testing data in HIV positives to be used for the enhanced estimator, which represents a more realistic scenario where not everyone has a prior test result.


\begin{figure}[htbp]
	\caption{Diagram of PC12 and PC24 data in the HPTN 071 PopART Study used for cross-sectional and longitudinal incidence estimation. Dark grey boxes represent the data used for constructing the RITA and enhanced RITA. PC12 is the 12-month time point, and PC24 is the 24-month time point since study initiation. $(Q_i', Q_i'T_i', Q_i'\Delta_i')$ denotes self-reported versions of $(Q_i, Q_iT_i, Q_i\Delta_i)$.}\label{fig:hptn-scheme}
	\centering
	\includegraphics[width=0.8\textwidth]{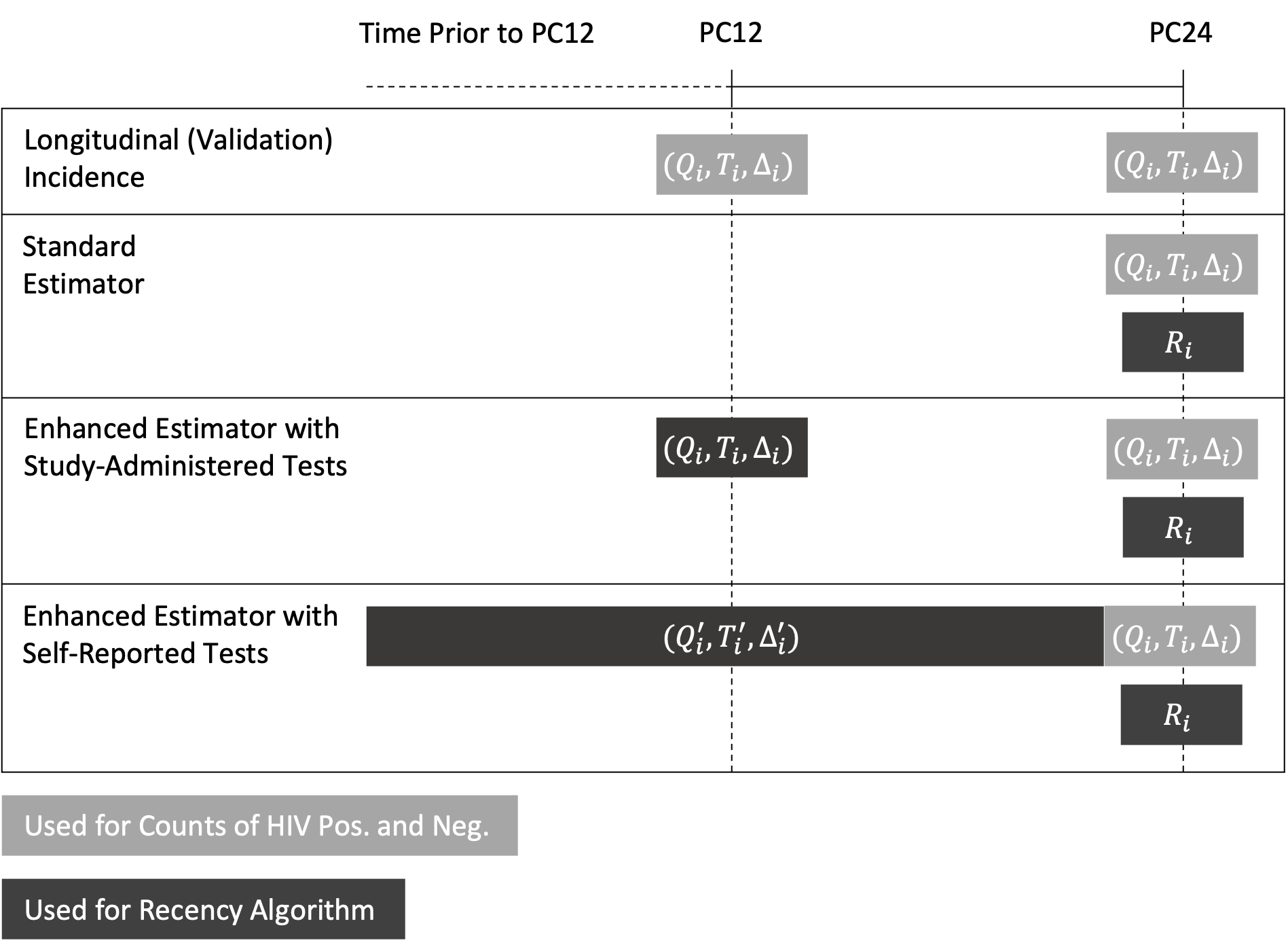}
\end{figure}

\subsection{CEPHIA Data for Estimation of RITA Characteristics}

The use of several RITAs (also referred to as multi-assay algorithms [MAAs] in cited publications) have been explored with PopART data \citep{klockValidationPopulationlevelHIV12021, grant-mcauleyEvaluationMultiassayAlgorithms2021}. We use the RITA that classifies someone as recently infected if their LAg ODn is less than 1.5 and their viral load is $> 1000$ copies/mL. To use the standard estimator, one only needs to extract an estimate of the MDRI for this RITA from relevant literature. As previously noted, this is not sufficient for the enhanced estimator: in addition to using an estimate of MDRI $\Omega_{T^*}$ for the RITA-recent algorithm, we also need its entire $\phi(t)$ test-recent function so that we can augment the estimator with prior test results (see equation \ref{est}). Here we estimate the $\phi(t)$ function using a public dataset from the Consortium for the Evaluation and Performance of HIV Incidence Assays (CEPHIA)'s evaluations of several HIV recency assays \citep{grebeCEPHIAPublicUse2021}. This dataset includes days post seroconversion, viral load, LAg-ODn, and HIV subtype. We use HIV subtype C since HIV infections in Zambia and South Africa are predominantly made up of subtype C \citep{hemelaarGlobalRegionalMolecular2019}. There are several functional forms that we could use to estimate the $\phi(t)$ function (e.g., different degree polynomial functions of either $U_i$ or $\log(U_i)$). Ultimately, they all have a very similar shape until 2 years (and resulting MDRI), which is all that is used for the enhanced estimator. We chose to estimate $\phi(t)$ using a 2nd degree polynomial on $\log(U_i)$.

We will estimate incidence using a 2-year definition of recent infection. Our estimate of MDRI on the CEPHIA data for 2 years is 160 days (95\% CI: 141 - 179). We use an FRR of 0\%, which was done in similar work \citep{klockValidationPopulationlevelHIV12021}. This is a reasonable assumption since our estimated test-recent function past 2 years infection duration is essentially zero. Only one person in the CEPHIA dataset was a false-recent leading to an observed FRR in CEPHIA data of about 0.1\%.

\subsection{Results}

\begin{figure}[htbp]
	\caption{Incidence estimates for the 21 communities in the HPTN 071 cohort at PC24. The red line is the standard estimator. The green line is the longitudinal incidence estimator (cases / person-time). The solid blue line is the enhanced estimator based on HPTN 071 study-administered prior HIV tests, and the dotted blue line is the same enhanced estimator, but excluding prior HIV tests for a randomly selected 50\% of HIV-positive individuals. The table below the plot includes: N Pos, the number of individuals identified as HIV positive at PC24 ($N_{pos}$); N Rec (RITA), the number identified as RITA-recent ($N_{rec}$); N Rec (Sero), the number of individuals who we know are recently infected from their prior HIV test because they seroconverted after time $t_{cs} - T^*$; N Rec (PT-RITA), the number identified as PT-RITA-recent ($N_{rec}^{PT}$). Confidence intervals created on the log incidence scale and exponentiated. Communities are ordered by longitudinal incidence estimate (lowest incidence on the left).}
	\centering
	\includegraphics[width=\textwidth]{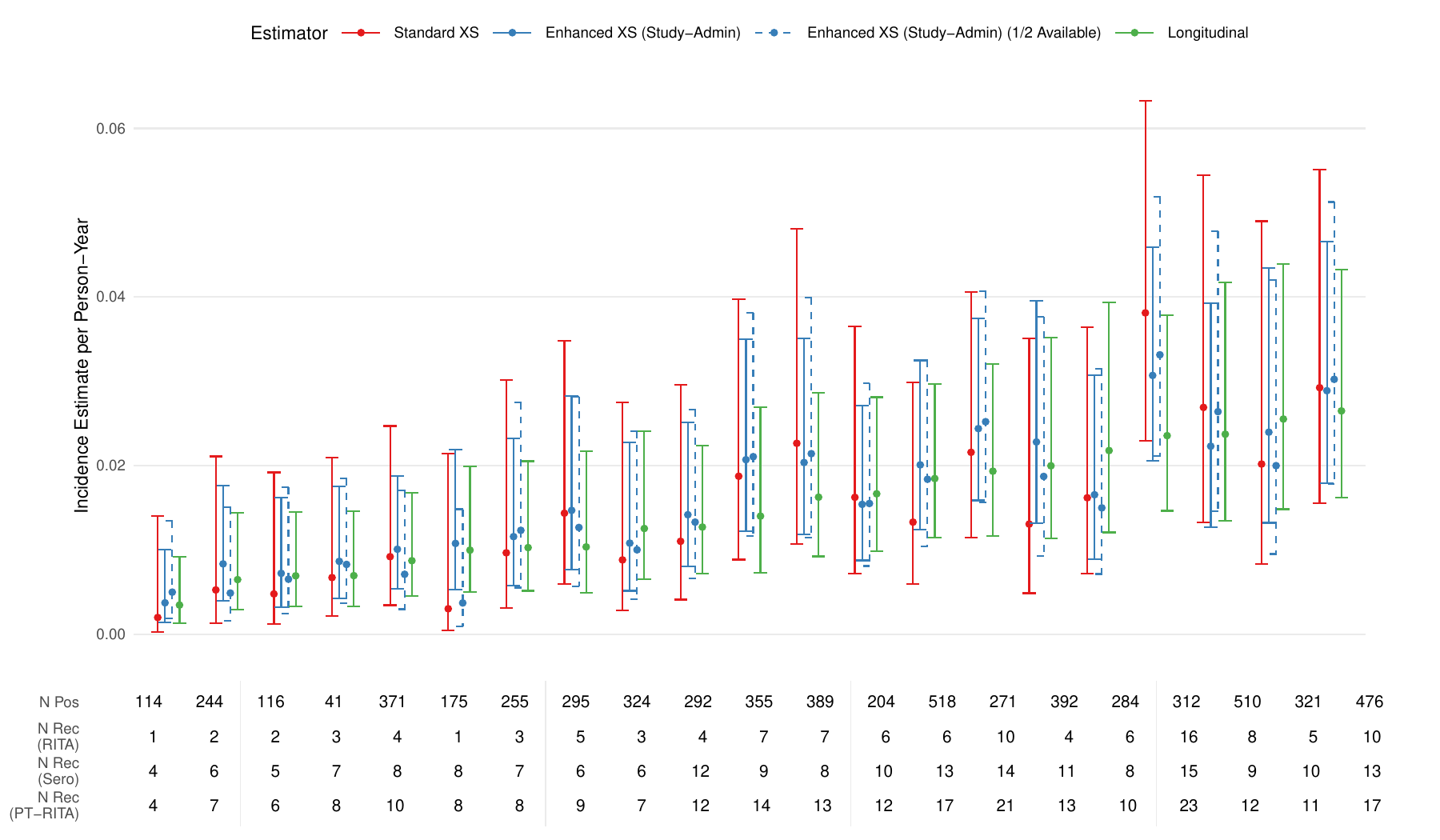}\label{fig:hptn-results}
\end{figure}

Figure \ref{fig:hptn-results} shows the estimated incidence at 24 months using the longitudinal estimator, the standard estimator, and the enhanced estimator with study-administered tests (including all of the prior HIV tests, and a randomly selected half of the prior HIV tests in HIV-positive individuals). Both the standard and enhanced estimators appear unbiased. In all communities, the enhanced point estimate falls within the confidence interval of the longitudinal estimator. In all but two communities, the same is true for the standard estimator. The confidence interval width for the enhanced estimator is usually a lot shorter than that of the standard estimator and is comparable to that of the longitudinal estimator.
The average variance reduction by using the enhanced estimator with study-administered HIV tests over the standard estimator across communities is 33\%. If we only include half of the prior testing data available, the variance is still reduced by 23\%.

We include counts of individuals identified as recent by the RITA and the PT-RITA below the main plot. All of the prior HIV tests occurred during the window defining recent infection, i.e., no individuals had $L_i = 1$. Thus, the total number of individuals with $R_i^* = 1$ is $N_{rec}^{PT} - N_{rec}$. This is the number of individuals who did not test RITA-recent, but who had a prior HIV test indicating recent infection and are therefore \textit{newly} identified by the PT-RITA.

Since the enhanced estimator and longitudinal estimator share a lot of data in this particular data analysis, we did an additional analysis where we randomly split the data 50-50 for each community into training and validation samples. The standard and enhanced estimators used data from the training sample, whereas the longitudinal estimator uses the validation sample. The results for this sample splitting analysis are similar to the main analysis (see Web Appendix D Figure A2).

In the study, participants are asked about the timing and results of the last HIV test that they have received at the 24 month visit.
These records are kept separately from the study-administered HIV testing results, so they provide an opportunity to check the validity of self-reported data.
Web Appendix D Figure A1 shows the estimated incidence at 24 months using the enhanced estimator with self-reported testing results. 
Using self-reported tests results in a large amount of positive bias in the incidence estimates, suggesting the possibility that some HIV-positive individuals with positive prior tests actually report negative prior tests (or no result at all), based on the results of our sensitivity analysis.

\section{Discussion}
Utility of cross-sectional incidence estimation using a recency testing algorithm is often limited by its power and precision.
In this paper, we consider incorporation of prior HIV test results to the cross-sectional incidence estimation and demonstrated that it has the potential to dramatically increase the precision of these estimates if prior testing results are available for those who test HIV positive. Even when only some HIV-positive individuals (e.g., 20\%) have prior tests available, the precision of the enhanced estimator is still substantially increased over the standard estimator.
In particular, increase in precision is mostly driven by the proportion of HIV-positive individuals with prior HIV-negative test results within the recent infection window.

With the increase in precision comes the reliance on additional assumptions on both the underlying HIV epidemic dynamics as well as the data generating process for prior HIV tests.
Rather than requiring constant incidence during the period prior to the cross-sectional survey that defines a recent infection (e.g., 2 years) as for the standard estimator, the enhanced estimator requires constant incidence strictly after the oldest prior test result used.
If prior test results were taken a long time ago, then the assumptions of the enhanced estimator may be more readily violated.
This scenario can be mitigated if non-recent prior test results are ignored.
And since the recency assay already has a low false recency rate (e.g., 1.4\% as we have in simulations), the gains from using the enhanced estimator come primarily from the recent prior HIV-negative tests, not the non-recent results.

The assumption on the data generating process of the prior test results (Assumption \ref{as:independence-testing}) is key to the validity of the enhanced estimator.
In our sensitivity analyses, we found substantially increased bias for the proposed enhanced estimator under violation of Assumption \ref{as:independence-testing}, when there could be two possible testing mechanisms: routine surveillance for HIV; and testing that is a function of how long the person has been infected. 
In future work, we will aim at relaxing Assumption \ref{as:independence-testing}, possibly through a propensity score approach accounting for important factors that affect individual's testing behavior, e.g., age, gender, and education level \citep{maheu-girouxDeterminantsTimeHIV2017}.

A criticism of cross-sectional incidence estimators in general is that individuals who have previously tested positive for HIV may not enter into the cross-sectional sample, which is a form of selection bias. This is a similar mechanism that leads to a violation of Assumption \ref{as:independence-testing}. Similar to what we have discussed in the prior paragraph, future work could leverage the prior testing data framework that we have introduced in this paper and with associated covariates to identify whether there is such selection bias and potentially correct for it.

Another practical issue that may lead to a violation of the required assumptions is the use of antiretroviral therapy (ART). If an individual receives a positive HIV test before time $t_{cs}$, they would likely be placed on ART, and as a result be virally suppressed. If the RITA at time $t_{cs}$ uses viral load, that individual would be less likely to be classified as recently infected \citep{grant-mcauleyEvaluationMultiassayAlgorithms2022}. In this way, a prior positive HIV test could influence the RITA result, violating Assumption \ref{as:cond-indep}. There are RITAs that do not use viral load \citep{grant-mcauleyEvaluationMultiassayAlgorithms2022}, which may be preferred for use with the PT-RITA approach to ensure that Assumption \ref{as:cond-indep} holds.

Application of the proposed enhanced estimation approach relies heavily on correct reporting of HIV testing results in HIV-positive individuals in the cross-sectional sample.
In our sensitivity analyses, we found that incorrect reporting of HIV test results would lead to substantial bias, especially when a prior HIV-positive result is reported to be HIV-negative. We see evidence of such bias in our analysis of the HPTN 071 data. When study-administered prior HIV tests were used in the enhanced estimator, the precision of the cross-sectional incidence estimates was substantially improved. When self-reported prior HIV tests reported at the same time point were used with the same recency test results, these incidence estimates had substantial bias.
Therefore, we recommend the proposed enhanced approach only with prior HIV test results conducted and recorded through a reliable source, e.g., routine clinic visits or infectious disease surveillance systems.

In addition to determining the effectiveness of population-based interventions as discussed in the introduction, cross-sectional incidence estimation with recency testing algorithm has also been suggested to estimate the incidence in a counterfactual placebo arm of an active-controlled trial \citep{gao2020sample,RAWG}, in order to assess efficacy of the trial products relative to placebo.
The recency testing algorithm can be applied to HIV-positive individuals in the screening population so that the cross-sectional incidence estimator may represent a concurrent placebo incidence. The validity of this interpretation requires several important assumptions, in addition to those required for the standard estimator \citep{gao2021statistical}. For example, there can be no PrEP use in the screening population, and no systematic differences between HIV risks of those who test HIV-positive and eligible HIV-uninfected persons who are enrolled in the trial \citep{gao2020sample}.
The proposed enhanced estimator is also applicable in such a setting to improve precision of incidence estimation.

In this paper, prior HIV testing results are represented as $(Q_i, Q_iT_i, Q_i\Delta_i)$, which means we can incorporate at most one prior HIV test result per person for the enhanced estimator. In settings where HIV tests are administered frequently (such as for populations enrolled in longitudinal household surveys or under routine surveillance for HIV infection), more individuals will have prior test results available, and individuals are more likely to have multiple test results.
Future work may extend the proposed methods to incorporate multiple tests per individual.
Additionally, current work leverages specific types of prior HIV testing results that would determine or exclude an individual as being recent.
Other testing results, for example an HIV-positive result less than $T^*$ time ago, are not utilized.
Future work should explore how to further incorporate such data to improve the enhanced algorithm.


\section*{Acknowledgements}

This work was supported by the U.S. National Institutes of Health grants R01AI177078, R01AI029168, and UM1AI068617.
Scientific computing at the Fred Hutch is supported by ORIP (S10OD028685). Additional support was provided by the Division of Intramural Research, National Institute of Allergy and Infectious Diseases, NIH.\vspace*{-8pt}

\section*{Data Availability}

Functions that implement the standard and enhanced estimators are available in the \texttt{XSRecency} package at \url{https://github.com/mbannick/XSRecency}. Simulations, figures, and tables can be reproduced using the simulation code here \url{https://github.com/mbannick/RITA-plus-sims}. The Public Use CEPHIA data from \cite{grebe2019post} that we use in our data analysis are available here \url{https://zenodo.org/record/4900634}.



\bibliographystyle{biom} \bibliography{bib}

\label{lastpage}

\end{document}


\appendix

\title{\setstretch{1.5}Supporting Information for: An Enhanced Cross-Sectional HIV Incidence Estimator that Incorporates Prior HIV Test Results}
\setstretch{1}
\author{Marlena Bannick, Deborah Donnell, Richard Hayes, Oliver Laeyendecker, Fei Gao}
\maketitle
\setcounter{tocdepth}{2}
{\sffamily\tableofcontents}
\doublespacing

\section{Characteristics of the Enhanced Algorithm}\label{app:characteristics}
The MDRI of PT-RITA is given by $\Omega_{T^*} = \int_0^{T^*}\phi(u)du$. Therefore, the MDRI of the new algorithm is:
\begin{align}
    \begin{split}
    \Omega_{T^*}^{PT} 
    &= \int_0^{T^*} \Pr(R_i^{PT}=1|U_i=u) du \\
      &= \int_0^{T^*}\Pr(R_i=1,L_i=0|U_i = u)du + \int_0^{T^*}\Pr(R_i=0, R_i^* = 1|U_i = u)du
    \end{split} \\
    \begin{split} \nonumber
     &= \int_0^{T^*}\Pr(R_i=1|U_i = u)du + \int_0^{T^*}\Pr(R_i=0, R_i^* = 1|U_i = u)du \\
     &= \int_0^{T^*}\phi(u)du + \int_0^{T^*}\{1-\phi(u)\}\Pr(T_i \leq T^*, \Delta_i=0, Q_i=1|U_i = u)du \quad \text{(Assumption \ref{as:cond-indep})}\\
      &= \int_0^{T^*}\phi(u)du + \int_0^{T^*}\{1-\phi(u)\}E\Big\{E(Q_i (1-\Delta_i)I(T_i \leq T^*)|U_i=u, Q_i, T_i)\Big|U_i=u\Big\} du \\
      &= \int_0^{T^*}\phi(u)du + \int_0^{T^*} \{1-\phi(u)\} E \Big\{Q_i I(T_i \leq T^*)E(1-\Delta_i|U_i=u, Q_i, T_i) \Big|U_i=u\Big\} du\\
      &= \int_0^{T^*}\phi(u)du +  \int_0^{T^*} \{1-\phi(u)\} E \Big\{Q_i I(T_i \leq T^*)I(u < T_i) \Big\} du \quad \text{(Assumption \ref{as:independence-testing})} \\
      &= \int_0^{T^*}\phi(u)du + E \Big\{Q_i I(T_i \leq T^*)\int_0^{T_i}\{1-\phi(u)\} du\Big\}
      \end{split} \\
      \begin{split}
      &= \Omega_{T^*} + E \Big\{Q_i I(T_i \leq T^*) \int_0^{T_i}\{1-\phi(u)\} du\Big\} 
      .\label{eq:mdri-2}
      \end{split}
\end{align}
Let $G$ be the distribution function of infection durations in the tested population, and let $\beta_{T^*}$ be the FRR of the original RITA. The FRR of the PT-RITA can then be expressed as
\begin{align*}
    \beta^{PT}_{T^*} &= \frac{\int_{T^*}^{\infty} \Pr(R_i^{PT} = 1|U_i=u) G(du)}{\int_{T^*}^{\infty} G(du)}.
\end{align*}
Focusing on the integrand in the numerator, and for $u > T^*$, we have 
\begin{align*}
\Pr(R_i^{PT} = 1|U_i=u) &= \Pr((R_i=1, L_i=0) \text{ or } R_i^* = 1|U_i=u) \\ 
    &= \Pr(R_i=1,L_i=0|U_i = u) + \Pr(R_i=0, R_i^* = 1|U_i = u) \\
    &= \Pr(R_i=1,L_i=0|U_i = u),
\end{align*}
since $\Pr (R_i^* = 1) = 0$ when $U_i > T^*$, i.e., someone infected $u > T^*$ units ago cannot test negative at time $T_i \leq T^*$. Using the independence of $R_i$ and $(Q_i, T_i)$ conditional on $U_i$ from Assumption \ref{as:cond-indep}, and the fact that $\Delta_i$ is deterministic in the presence of $T_i$ once we have conditioned on $U_i = u$, we have:
\begin{align*}
    \Pr(R_i=1,L_i=0|U_i = u) &= \Pr(R_i=1|U_i=u) - \Pr(R_i=1,L_i=1|U_i=u) \\
    &= \Pr(R_i=1|U_i=u) - \Pr(R_i=1|U_i=u) \Pr(T_i > T^*, \Delta_i=1, Q_i=1|U_i=u) \\
    &= \phi(u) \left\{1 - \Pr(T_i > T^*, \Delta_i=1, Q_i=1|U_i=u) \right\} \\
    &= \phi(u) \left\{1 - \Pr(T^* < T_i \leq u, \Delta_i=1, Q_i=1|U_i=u) \right\}.
\end{align*}
Therefore,
\begin{align*}
    \beta^{PT}_{T^*} = \frac{\int_{T^*}^{\infty} \phi(u) \big\{1 - \Pr(T^* < T_i \leq u, \Delta_i=1, Q_i=1|U_i=u) \big\} dG(u)}{\int_{T^*}^{\infty}dG(u)} \leq \frac{\int_{T^*}^{\infty} \phi(u) dG(u)}{\int_{T^*}^{\infty}dG(u)} \equiv \beta_{T^*}.
\end{align*}
\section{The Enhanced Estimator}

\subsection{Estimator Derivation}\label{app:estimator}

Let $p(s) = \Pr(I_i \leq s)$. In the main text, $s = t_{cs}$. Let incidence be
\begin{align*}
    \lambda(s) = \lim_{ds\to0}\frac{1}{ds} \Pr(s\le I_i<s+ds | I_i\ge s).
\end{align*}
For simplicity of notation we will write $\lambda(s) = \Pr(I_i=s|I_i\ge s)$. Then, the PT-RITA-recent probability for HIV-positive subject $i$ at time $s$ is given by
\begin{align*}
    P_{rec,i}^{PT}(s) &= \Pr(R_i^{PT}=1|I_i\le s)\\
    &=\int_{0}^{\infty} \Pr(R_i^{PT} = 1, I_i = s - u|I_i\le s) du \\
    &= \int_{0}^{\infty} \Pr(R_i^{PT} = 1|I_i = s - u)\Pr(I_i = s - u|I_i\le s) du \\
    &= \int_{0}^{\infty} \Pr(R_i^{PT} = 1 | U_i = u) \Pr(I_i = s - u|I_i\le s) du.
\end{align*}
Note that we can rewrite the second probability above by Assumption \ref{as:constant-inc}, giving it a short-hand notation to simplify the derivations:
\begin{align}
    P_i(s, u) := \Pr(I_i = s - u | I_i \leq s) &= \Pr(I_i = s| I_i \leq s)\nonumber \\
    &=
    \frac{\Pr(I_i = s | I_i \geq s) \Pr(I_i \geq s)}{\Pr(I_i \leq s)} \label{eq:incidence} \\
    &= \frac{\lambda(s)(1-p(s))}{p(s)}.\nonumber
\end{align}
Notice that $\Pr(R_i^{PT} = 1 | U_i = u)$ can be broken up into three terms based on the definition of $R_i^{PT}$, so we now have
\begin{align*}
    P_{rec,i}^{PT}(s) &= \int_{0}^{\infty} \Pr(R_i^{PT} = 1 | U_i = u) P_i(s, u) du \\
    &= \int_{0}^{\infty} \big\{\Pr(R_i = 1, L_i = 0 | U_i = u) + \Pr(R_i^* = 1 | U_i = u)\big\} P_i(s, u) du \\
    &= \underbrace{\int_{0}^{\infty} \Pr(R_i = 1 | U_i = u) P_i(s, u) du}_{(I)} - \underbrace{\int_{0}^{\infty} \Pr(R_i = 1, T_i \geq T^*, \Delta_i=1, Q_i=1 | U_i = u) P_i(s, u) du}_{(II)}\\
    &\qquad+\underbrace{\int_{0}^{\infty} \Pr(R_i = 0,T_i\le T^*,\Delta_i=0, Q_i=1 | U_i = u) P_i(s, u) du}_{(III)}.
\end{align*}
For term (I), we have
\begin{align}
    \int_{0}^{\infty} \Pr(R_i = 1 | U_i = u) P_i(s, u) du &= \int_0^{T^*} \Pr(R_i = 1 | U_i = u) P_i(s, u) du + \int_{T^*}^{\infty} \Pr(R_i = 1 | U_i = u) P_i(s, u) du \nonumber \\
    &= \int_0^{T^*} \phi(u) P_i(s, u) du + \beta_{T^*} \Big\{1-\int_{0}^{T^*} P_i(s, u) du\Big\} \nonumber \\
    &= \beta_{T^*} + \big\{\Omega_{T^*} - \beta_{T^*}T^*\big\} \frac{\lambda(s)(1-p(s))}{p(s)}\label{eq:prob-recent}.
\end{align}
Similarly, for term (II), we have,
\begin{align*}
    &\int_{0}^{\infty} \Pr(R_i = 1, T_i \geq T^*, \Delta_i=1, Q_i=1 | U_i = u) P_i(s, u) du\\
    &=\int_{0}^{\infty} \Pr(R_i = 1 | U_i = u)\Pr(T_i \geq T^*, \Delta_i=1, Q_i=1 | U_i = u) P_i(s, u) du\\
    &=\int_{0}^{\infty} \phi(u)\int_{T^*}^\infty\Pr(T_i=t, \Delta_i=1, Q_i=1 | U_i = u) dt P_i(s, u) du\\
    &=\int_{0}^{\infty} \phi(u)\int_{T^*}^\infty\Pr(\Delta_i=1 |T_i=t, Q_i=1,  U_i = u)\Pr(T_i=t,Q_i=1|U_i=u) dt P_i(s, u) du\\
    &=\int_{0}^{\infty} \phi(u)\int_{T^*}^\infty I(t\le u)\Pr(T_i=t,Q_i=1) dt P_i(s, u) du\\
    &=\int_{0}^{\infty} \phi(u)\Pr(T^*\leq T_i\le u,Q_i=1)P_i(s, u) du\\
    &=E\left\{I(T^*\leq T_i,Q_i=1)\int_{T_i}^{\infty} \phi(u)P_i(s, u) du\right\}\\
    &=E\left[I(T^*\leq T_i,Q_i=1)\beta_{T^*}\left\{1-\int_{0}^{T_i} P_i(s, u) du\right\}\right]\\
    &=E\left[I(T^*\leq T_i,Q_i=1)\beta_{T^*}\left\{1-T_i\frac{\lambda(s) (1-p(s))}{p(s)}\right\}\right].
\end{align*}
Finally, for term (III), we have
\begin{align*}
    \int_{0}^{\infty} \Pr(R_i = 0,T_i\le T^*,\Delta_i=0,Q_i=1 | U_i = u) P_i(s, u) du = E\left[Q_iI(T_i\le T^*)\frac{\lambda(s) (1-p(s))}{p(s)} \int_{0}^{T_i} \{1-\phi(u)\} du\right].
\end{align*}
Then,
\begin{align*}
    P_{rec,i}^{PT}(s) &= \Omega_{T^*}\frac{\lambda(s) (1-p(s))}{p(s)} + \beta_{T^*}\left\{1-T^*\frac{\lambda(s) (1-p(s))}{p(s)}\right\} - E\left[Q_iI(T_i\geq T^*)\beta_{T^*}\left\{1-T_i\frac{\lambda(s) (1-p(s))}{p(s)}\right\}\right]\\
    &\qquad+ E\left[Q_iI(T_i\le T^*)\int_{0}^{T_i} \{1-\phi(u)\} du\frac{\lambda(s) (1-p(s))}{p(s)}\right]\\
	&=\frac{\lambda(s) (1-p(s))}{p(s)}\left\{\Omega_{T^*}  + \E \big\{Q_iI(T_i\le T^*)\int_{0}^{T_i} \{1-\phi(u)\} du\big\} - \beta_{T^*}\big(T^* - \E\big\{Q_iI(T_i\geq T^*)T_i\big\}\big)\right\}\\
	&\qquad+\beta_{T^*}(1-\E\{Q_iI(T_i\geq T^*)\}).\label{prec}
\end{align*}
We can then formulate an incidence estimator for $\lambda(s)$ by solving for $\lambda(s)$,
\begin{equation}
	\lambda(s) = \Bigg(\frac{p(s)}{1-p(s)}\Bigg)\frac{\Big\{P^{PT}_{rec}(s) - \beta_{T^*}(1-\E\{Q_iI(T_i\geq  T^*)\}) \Big\}}{ \Big\{\Omega_{T^*}  + \E \big\{Q_iI(T_i\le T^*)\int_{0}^{T_i} \{1-\phi(u)\} du\big\} - \beta_{T^*}\big(T^* - \E\big\{Q_iI(T_i\geq T^*)T_i\big\}\big)  \Big\}},\label{eq:incidence-param}
\end{equation}
and plugging in estimators for all unknown quantities (and recalling that the sample data and prior data are independent from the data used to create estimators of $\Omega_{T^*}$, $\beta_{T^*}$, and $\phi$),
\begin{align*}
    \hat\lambda &= \frac{(N_{pos}/N) \Big\{N_{rec}^{PT}/N_{pos} - (1/N_{pos})\sum_{i:D_i=1}\hat{\beta}_{T^*}\left[1-I(T_i \geq  T^*)Q_i\right]\Big\}}{(N_{neg}/N)\left\{\hat{\Omega}_{T^*}  + \sum_{i:D_i=1}\left[Q_i \cdot I(T_i\le T^*)\int_{0}^{T_i} \{1-\hat{\phi}(u)\} du - \hat{\beta}_{T^*}\{T^* - Q_i \cdot I(T_i \geq T^*)T_i\}\right]/N_{pos}\right\}} \\
    &= \frac{N_{rec}^{PT} - \sum_{i:D_i=1}\hat{\beta}_{T^*}\left[1-I(T_i \geq T^*)Q_i\right]}{N_{neg}\left\{\hat{\Omega}_{T^*}  + \sum_{i:D_i=1}\left[Q_i \cdot I(T_i\le T^*)\int_{0}^{T_i} \{1-\hat{\phi}(u)\} du - \hat{\beta}_{T^*}\{T^* - Q_i \cdot I(T_i\geq T^*)T_i\}\right]/N_{pos}\right\}}. \\
\end{align*}

\subsection{Consistency of the Enhanced Estimator}

Under Assumptions \ref{as:constant-frr} through \ref{as:consistency}, we can show that $\hat{\lambda}$ is consistent for $\lambda(s)$, where $s$ is the calendar time of the cross-sectional sample. Noting the form of $\hat{\lambda}$, we see that almost each of the terms are sample statistics and $\hat{\lambda}$ is a plug-in estimator for $\lambda(s)$. Therefore, we would automatically have consistency of $\hat{\lambda}$ for $\lambda(s)$ by Slutsky's Lemma.

However, there is one term that is a more complicated, and that requires Assumption \ref{as:consistency}: $\sum_{i:D_i=1} [Q_i I(T_i \leq T^*) \int_0^{T_i} \{1 - \hat{\phi}(u)\}] du$. Consistency of this term is more complicated because of the plug-in estimate of $\hat{\phi}$, so this is not a simple average. Lemma \ref{lemma:consistency} allows us to conclude consistency of this term, which means that the entire $\hat{\lambda}$ is consistent.

\begin{lemma}\label{lemma:consistency}
    Let $\psi_n := \sum_{i:D_i=1} [Q_i I(T_i \leq T^*) \int_0^{T_i} \{1 - \hat{\phi}(u)\}] du$, and $\psi_0 = \E(Q_i 1(T_i \leq T^*) \int_0^{T_i} [1 - \phi(u) ] du$. Under Assumption \ref{as:consistency}, $\psi_n \to \psi_0$ in probability.
\end{lemma}
\begin{proof}
    Our expectations over $(Q_i, T_i)$ are only with respect to HIV positive individuals ($D_i = 1$). Since HIV positive individuals are a subset of a random cross-sectional sample, for the sake of simplicity, redefine $n$ as the number of HIV-positive individuals, and let $\psi_n := \frac1n \sum_{i=1}^{n} [Q_i I(T_i \leq T^*) \int_0^{T_i} \{1 - \hat{\phi}(u)\}] du$. Let $P_n(\cdot) = \frac{1}{n} \sum_{i=1}^{n}(\cdot)$ and $P_0(\cdot) = \int(\cdot) d\tilde{P}(Q_i, T_i)$, where $\tilde{P}(Q_i, T_i)$ is the probability measure for $(Q_i, T_i)$. Define the functions
    \begin{align*}
        f_n(Q_i, T_i) &:= Q_i 1(T_i \leq T^*) \int_0^{T_i} [1-\hat{\phi}(u)] du \\
        f_0(Q_i, T_i) &:= Q_i 1(T_i \leq T^*) \int_0^{T_i} [1-\phi(u)] du.
    \end{align*}
    We use the following expansion:
    \begin{align*}
        \psi_n - \psi_0 \equiv P_n f_n - P_0 f_0 = (P_n - P_0) f_0 + P_n (f_n - f_0).
    \end{align*}
    It suffices to show that both of the terms on the RHS above converge to 0 in probability. For $(P_n - P_0) f_0$, this is evident by the law of large numbers, since $f_0$ is a fixed function. For the second term, we have
    \begin{align*}
        |P_n (f_n - f_0)| &= \left|\frac1n \sum_{i=1}^{n} Q_i 1(T_i \leq T^*) \int_0^{T_i} [\hat{\phi}(u) - \phi(u)] \right| \\
        &\leq \frac1n \sum_{i=1}^{n} \int_0^{T_i} \left|\hat{\phi}(u) - \phi(u) \right| du \\
        &\leq \int_0^{\tau} \left|\hat{\phi}(u) - \phi(u) \right| du \\
       \text{(Assumption \ref{as:consistency})} \quad &\quad\quad \to 0 \text{ in probability}.
    \end{align*}
    Therefore, $\psi_n \to \psi_0$ in probability.
\end{proof}

\subsection{Shadow Period of the Enhanced Estimator}

In this section, we derive the mean shadow period \citep{kaplan1999snapshot} of the enhanced estimator, which indicates the time at which the estimate is consistent for the incidence, when the incidence is linearly changing in time.
Note that $\hat{\lambda} \to \lambda^*$ in probability, where $\lambda^*$ has the form in \eqref{eq:incidence-param}.
We wish to write the form of $\lambda^*$ when the true incidence function is
\begin{align*}
    \lambda(s-u) = \lambda(s) + \rho u
\end{align*}
where $s$ is a calendar time, and $u \in [0, \max(\tau, T^*)]$.
Here, Assumption \ref{as:constant-inc} is violated such that instead of equation \ref{eq:incidence} we have $P(s, u) = \lambda(s-u)(1-p(s))/p(s)$.
Using terms (I), (II), and (III) we have
\begin{align*}
    P_{rec}^{PT}(s) &:= \int_0^{T^*} \left\{\phi(u) - \beta_{T^*} \right\} P(s, u) du + \beta_{T^*}
     - \beta_{T^*} \E\left\{I(T_i \geq T^*, Q_i=1) \left(1 - \int_0^{T_i} P(s, u) du \right) \right\} \\
    &\quad + \E\left\{I(T_i \leq T^*, Q_i=1) \int_0^{T_i} \left(1 - \phi(u) \right) P(s, u) du \right\} \\
    &= \left(\frac{1-p(s)}{p(s)}\right)\int_0^{T^*} \left\{\phi(u) - \beta_{T^*} \right\} \lambda(s-u) du + \beta_{T^*} \\
    &\quad - \beta_{T^*} \E\left\{I(T_i > T^*, Q_i=1) \right\} + \beta_{T^*} \left(\frac{1-p(s)}{p(s)}\right) \E\left\{I(T_i \geq T^*, Q_i=1) \int_0^{T_i} \lambda(s-u) du \right\} \\
    &\quad + \left(\frac{1-p(s)}{p(s)}\right) \E\left\{I(T_i \leq T^*, Q_i=1) \int_0^{T_i} \left(1 - \phi(u) \right) \lambda(s-u) du \right\} \\
\end{align*}
We can use this to study $\lambda^*$, and derive its value when $\lambda(s-u) = \lambda(s) + \rho u$, by plugging in $P_{rec}^{PT}(s)$ into equation \eqref{eq:incidence-param}. For short-hand, let $A_i = Q_i I(T_i \leq T^*)$ and $B_i = Q_i I(T_i \geq T^*)$ be indicators for having an available recent or non-recent prior test. Then,
\begin{align*}
    \lambda^* &= \left(\frac{p(s)}{1-p(s)}\right)\frac{\left\{P_{rec}^{PT}(s) - \beta_{T^*}(1-\E\{B_i\}) \right\}}{\left\{\Omega_{T^*}  + \E \left\{A_i\int_{0}^{T_i} \{1-\phi(u)\} du\right\} - \beta_{T^*}\left(T^* - \E\left\{B_iT_i\right\}\right)  \right\}} \\
    &= \frac{\int_0^{T^*} \left\{\phi(u) - \beta_{T^*} \right\} \lambda(s - u) du + \E\left\{A_i \int_0^{T_i} \left(1 - \phi(u) \right) \lambda(s-u) du\right\} - \beta_{T^*} \E\left\{B_i \int_0^{T_i} \lambda(s-u) du \right\}}{\left(\Omega_{T^*} -\beta_{T^*} T^*\right) + \E \left\{A_i\int_{0}^{T_i} \{1-\phi(u)\} du\right\} + \beta_{T^*}\E\left\{B_iT_i\right\}}.
\end{align*}
Since each term of the numerator has $\lambda(s-u)$, we can linearly separate it into the numerator times $\lambda(s)$ outside of the integrals and the numerator times $\rho u$ inside the integrals. The integral of the numerator after pulling out $\lambda(s)$ cancels out with the denominator, so we are left with:
\begin{align*}
    \lambda^* &= \lambda(s) + \rho \frac{\left[\int_0^{T^*} \left\{\phi(u) - \beta_{T^*} \right\}u du \right] + \E\left\{A_i \int_0^{T_i} \left(1 - \phi(u) \right) u du\right\} - \beta_{T^*} \E\left\{B_i \int_0^{T_i} u du \right\}}{\left[\Omega_{T^*} -\beta_{T^*} T^*\right] + \left[\E \left\{A_i\int_{0}^{T_i} \{1-\phi(u)\} du\right\} + \beta_{T^*}\E\left\{B_iT_i\right\}\right]}
\end{align*}
We now simplify the second and third expression in the numerator of the term multiplied by $\rho$. Recall that $T_i$ is supported on $[0, \tau]$.
\begin{align*}
    \E\left\{B_i \int_0^{T_i} u du \right\}  &=\E\left\{Q_i I(T_i\geq T^*) \int_0^\tau I(u\le T_i) u du \right\}\\ 
     &=\int_0^\tau\E\left\{Q_i I(T_i \geq T^*)  I(u\le T_i)  \right\}u du
\end{align*}
where we have used Fubini's theorem to switch integral ordering. Similarly, we have
\begin{align*}
    \E\left\{A_i \int_0^{T_i} (1-\phi(u)) u du \right\}   &=\E\left\{Q_i I(T_i\le T^*) \int_0^\tau I(u\le T_i) (1-\phi(u)) u du \right\}\\
    &=\int_0^\tau\E\left\{Q_i I(T_i\le T^*)  I(u\le T_i)  \right\}(1-\phi(u)) u du\\
\end{align*}
Putting these terms together, we have
\begin{align*}
    \E \left\{A_i\int_{0}^{T_i} \{1-\phi(u)\} du\right\} + \beta_{T^*}\E\left\{B_iT_i\right\}
    &= \int_{0}^{\tau} u \E\left\{Q_i 1(u \leq T_i) \left[1(T_i \leq T^*) \left(1-\phi(u) \right) + 1(T_i \geq T^*) \beta_{T^*} \right] \right\} du.
\end{align*}
Therefore, $\lambda^* = \lambda(s) + \rho \omega^*$, where
\begin{align*}
    \omega^* := \frac{\left[\int_0^{T^*} \left\{\phi(u) - \beta_{T^*} \right\}u du \right] + \left[\int_{0}^{\tau} u \E\left\{Q_i 1(u \leq T_i) \left[1(T_i \leq T^*) \left(1-\phi(u) \right) + 1(T_i \geq T^*) \beta_{T^*} \right] \right\} du \right]}{\left[\Omega_{T^*} -\beta_{T^*} T^*\right] + \left[\E \left\{Q_i I(T_i \leq T^*) \int_{0}^{T_i} \{1-\phi(u)\} du\right\} + \beta_{T^*}\E\left\{Q_i I(T_i \geq T^*) T_i\right\}\right]}.
\end{align*}
The first term in the numerator, and the first term in the denominator, match that of the shadow period for the adjusted estimator \citep{gao2021statistical}. However, both the numerator and denominator have extra terms that depend on the distribution of prior testing data $(Q_i, T_i)$. The second term in the numerator is a weighted average of $u$'s over the support of $T_i \in [0, \tau]$, where the weights for $u$ are functions of the distribution of prior testing data. Similarly to the derivation in \citet{gao2021statistical}, $\hat{\lambda}$ is consistent for $\lambda(s - \omega^*)$, where $\omega^*$ is given above.
\section{Analytical Variance Derivation}\label{app:variance}
	
	Let $N$ be the total sample size, $N_{pos} \sim \text{Binomial}(N, p)$, be the number of people positive for HIV, and $N_{rec}^{PT} \sim \text{Binomial}(N_{pos}, P_{rec}^{PT})$ be the number of recents identified by the enhanced algorithm. We will also use $P_{rec}$, which is the probability of an HIV positive person being identified as recent based on the recency test alone that was used in the derivations of \cite{gao2021statistical}.
	
	Recall that $Q_i$ is an indicator for having a prior test result available. Let $q \equiv P(Q_i=1)$, thus we can think of $Q_i \sim \text{Binomial}(q)$. Let $c = P(T_i \leq T^* | Q_i=1)$, with $C_i := 1(T_i \leq T^* | Q_i = 1) \sim \text{Binomial}(c)$. Thus, $C_i$ and $Q_i$ are independent binomial random variables. Let $A_i = Q_i 1(T_i \leq T^*)$ and $B_i = Q_i 1(T_i \geq T^*)$, indicators for having an available recent or non-recent prior test.
 Let $N_{A} = \sum_{i:D_i=1} A_i$, and $N_{B} = \sum_{i:D_i=1} B_i$. Then based on the relationship with $C_i$ and $Q_i$, we have $N_{A} \sim \text{Binomial}(N_{pos}, cq)$ and $N_B \sim \text{Binomial}(N_{pos}, (1-c)q)$. For short-hand, we will call $P_A = cq$ and $P_B = (1-c)q$. With this new notation, we can write our incidence estimator from \eqref{est} as
	\begin{align*}
		\hat{\lambda} = \frac{W_1}{(N-W_2)(W_3 + [W_4 + W_5]/W_2)}
	\end{align*}
\begin{itemize}[label=--]
		\item $W_1 = N_{rec}^{PT} - \hat{\beta}_{T^*} \sum_{i:D_i=1} (1-B_i) = N_{rec}^{PT} - \hat{\beta}_{T^*} (N_{pos} - N_B)$
		\item $W_2 = N_{pos}$
		\item $W_3 = \hat{\Omega}_{T^*} - \hat{\beta}_{T^*} T^*$
		\item $W_4 = \sum_{i:D_i=1} A_i \big\{T_i-\int_0^{T_i} \hat{\phi}(u) du\big\}$
		\item $W_5 = \sum_{i:D_i=1} \hat{\beta}_{T^*} B_i T_i$
	\end{itemize}
\newpage
\subsection{Plug-In Terms}
		In the variance derivation, we will use several terms that simplify the resulting expressions for the variance and allow us to plug in empirical expectations and variances to estimate the overall variance of $\hat{\lambda}$:
	 \begin{itemize}
	 	\item We have $\phi(u) \equiv \E[\hat{\phi}(u)]$ and let $\rho(u, v) := \Cov[\hat{\phi}(u), \hat{\phi}(v)]$. This is the covariance of the estimated test-recent function $\hat{\phi}$ at different function evaluations $u, v > 0$.
	 	\item Define $\Omega_{t} := \int_0^{t} \phi(u) du$. Let $\omega_{T,A} := \E[\Omega_{T_i}|A_i=1]$, $\sigma^2_{\omega_{T,A}} := \Var[\Omega_{T_i}|A_i=1]$, and $\omega^*_{T,A} = \E[T_i \Omega_{T_i}|A_i=1]$. These are all quantities related to the mean duration of recent infection, but where the $T^*$ in the integration is the random variable $T_i$.
	 	\item Let $r_{t_i, t_j} = \int_0^{t_i} \int_0^{t_j} \rho(u, v) du dv$. Let $r_{T,A} = \E[r_{T_i, T_i}|A_i=1]$ and $r'_{T,A} = \E[r_{T_i, T_j}|A_i=1]$, and let $r^*_{T,A} = \E[r_{T_i,T^*}|A_i=1]$. These are the integrated covariance functions where the 2-dimensional grid of integration is a random variable \textit{among those with a recent prior test}.
	 	\item Let $\mu_{T,A} = \E(T_i|A_i=1)$ and $\sigma^2_{T,A} = \Var(T_i|A_i=1)$, the expected time (and its variance) of a prior test if it was less than $T^*$ ago. Let $\mu_{T,B} = \E(T_i|B_i=1)$ and $\sigma^2_{T,B} = \Var(T_i|B_i=1)$, be the same quantities, but for those whose prior test was more than $T^*$ ago.
	 \end{itemize}

\newpage
\subsection{Recency Algorithm in Indicator Form}
$R_i^{PT}$ is our recency indicator for the PT-RITA. When we compute the variance of $W_1$ (section \ref{var:w1}), and the covariance of $W_1$ and $W_5$ (section \ref{cov:15}), we will need to work with the components of $N_{rec}^{PT} = \sum_{i:D_i=1} R_i^{PT}$. Here we derive an equivalent form with the new notation of $(A_i, B_i)$:
\begin{align}
    R_i^{PT} &= R_i - L_i + R_i^* \nonumber\\
            &= R_i(1 - I(T_i \geq  T^*)\Delta_i Q_i) + (1-R_i) I(T_i \leq T^*) (1 - \Delta_i) Q_i \\
                            &= R_i(1 - B_i\Delta_i ) + (1 - R_i) A_i (1-\Delta_i).\label{mi}
\end{align}
$R_i$ is independent of $A_i$ and $B_i$ by Assumption \ref{as:independence-testing}; however, $\Delta_i$ is dependent with $R_i$, $A_i$ and $B_i$.

\subsection{Covariance Identities}
\begin{itemize}
	\item For any random variables $X, Y, Z$ where $X$ is independent of $Y$ and $Z$, we have the following covariance identity: $\Cov(XY, XZ) = \Var(X) \Cov(Y, Z) + \E(X)^2 \Cov(Y, Z) + \E(Y) \E(Z) \Var(X)$.
	\item For any two random variables $X_i, Y_i$, and i.i.d. data, we have that by the law of total covariance,
\begin{align*}
	\Cov\left(\sum_{i:D_i=1} X_i, \sum_{i:D_i=1} Y_i \right) &= \Cov(X_i, Y_i)\E\left(N_{pos}  \right) + \Var(N_{pos}) \E(X_i) \E(Y_i) \\
	&= N p\Big\{\E(X_iY_i) - p\E(X_i)\E(Y_i) \Big\}.
\end{align*}
Similarly, $\Var\left(\sum_{i:D_i=1}X_i\right) = N p \Var(X_i) + N p (1-p) \E(X_i)^2$.
	\item For any random variables $X, Y, Z$ if $X \indep Y$, then $\Cov(Z, XY) = \E(X) \Cov(Z, Y)$.
\end{itemize}

	 \newpage
	 \subsection{Calculation of Variance and Covariance Terms}
\begingroup
    \fontsize{10pt}{12pt}\selectfont
    We show the derivations for complicated terms that are not already included in \cite{gao2020sample}:
	 \begin{align*}
	 \E(W_1) &= N p \cdot \big\{P_{rec}^{PT} - \beta_{T^*} (1-P_B)\big\} \\
	 \E(W_2) &= Np\\
	 \E(W_3) &= \Omega_{T^*} - \beta_{T^*} T^* \\
	 \E(W_4) &= N p \cdot P_A \cdot (\mu_{T,A} - \omega_{T,A})\\
	 \E(W_5) &= N p \cdot \beta_{T^*} \cdot \E(B_i T_i) = N p \cdot \beta_{T^*} P_B \cdot \mu_{T,B} \\
	 \hyperref[var:w1]{\Var(W_1)} &= N p \Big\{P_{rec}^{PT}(1-P_{rec}^{PT}) + (1-p) (P_{rec}^{PT} - (1-P_B) \beta_{T^*})^2 \\
	&\quad\quad\quad + \sigma^2_{\hat{\beta}_{T^*}} (1-P_B)\Big\{1-(1-P_B)p + N p (1-P_B) \Big\} \\
	&\quad\quad\quad + \beta_{T^*}P_B\big\{\beta_{T^*}(1-P_B) - 2 \left(P_{rec}^{PT} - (P_{rec} - \beta_{T^*}) \left\{ 1 + \mu_{T,B}\beta_{T^*}/(\Omega_{T^*} - \beta_{T^*} T^*) \right\}\right) \big\} \Big\} \\	 \Var(W_2) &= Np(1-p)\\
	 \Var(W_3) &= \sigma^2_{\hat{\Omega}_{T^*}} + \sigma^2_{\hat{\beta}_{T^*}} T^{*2}\\
	 \hyperref[var:w4]{\Var(W_4)} &= N p P_A \big \{(\sigma^2_{T,A} + \sigma^2_{\omega_{T,A}}) + (\mu^2_{T,A} + \omega^2_{T,A}) + (1-p P_A) (\mu_{T,A} - \omega_{T,A})^2 \\
&\quad\quad\quad\quad\quad + r_{T,A} + P_Ar'_{T,A} (Np - p) - 2\omega_{T,A}^* \big\} \\
	 \hyperref[var:w5]{\Var(W_5)} &= N p P_B \Big\{\sigma^2_{\hat{\beta}_{T^*}} \mu^2_{T,B}P_B N p + \left(\sigma^2_{\hat{\beta}_{T^*}} + \beta^2_{T^*} \right) \left(\sigma^2_{T,B} + \mu_{T,B}^2(1-P_Bp) \right)\Big\}\\
	 \hyperref[cov:12]{\Cov(W_1, W_2)} &= N p (1-p) \big\{P_{rec}^{PT} -(1-P_B) \beta_{T^*} \big\} \\
	 \hyperref[cov:13]{\Cov(W_1, W_3)} &= N p \cdot T^* \sigma^2_{\hat{\beta}_{T^*}}  (1-P_B) \\
	 \hyperref[cov:14]{\Cov(W_1, W_4)} &= Np P_A \bigg\{(\mu_{T,A} - \omega_{T,A}) \big(P_{rec} - pP_{rec}^{PT} - \beta_{T^*}(1-p+pP_B) \big) \\
	&\quad\quad\quad\quad\quad + \frac{(P_{rec} - \beta_{T^*})}{(\Omega_{T^*} - \beta_{T^*} T^*)} \big(\mu_{T,A}^2 + \sigma_{T,A}^2 + \omega_{T,A}^2 + \sigma^2_{\omega_{T,A}} - 2\omega_{T,A}^*\big) \bigg\} \\
	 \hyperref[cov:15]{\Cov(W_1, W_5)} &= N p P_B \Big\{\beta_{T^*} \Big((P_{rec} - \beta_{T^*}) \Big(\mu_{T,B} + (\Omega_{T^*} - \beta_{T^*} T^*)(\sigma_{T,B}^2 + \mu^2_{T,B}) \beta_{T^*}\Big) - pP_{rec}^{PT} \mu_{T,B} \Big) \\
	&\quad\quad\quad\quad + p \Big(\sigma^2_{\hat{\beta}_{T^*}} + \beta_{T^*}^2 \Big) \mu_{T,B} (1-P_B) - \sigma^2_{\hat{\beta}_{T^*}} N p (1-P_B) \mu_{T,B}\Big\} \\
	 \hyperref[cov:24]{\Cov(W_2, W_4)} &= N p (1-p) \cdot P_A \cdot (\mu_{T,A} - \omega_{T,A})  \\
	 \hyperref[cov:25]{\Cov(W_2, W_5)} &= N p (1-p) \cdot \beta_{T^*} \cdot P_B \mu_{T,B} \\
	 \hyperref[cov:34]{\Cov(W_3, W_4)} &= -N p \cdot P_A \cdot r_{T,A}^* \\
	 \hyperref[cov:35]{\Cov(W_3, W_5)} &= - N p \cdot T^* \cdot P_B \mu_{T,B} \sigma^2_{\hat{\beta}_{T^*}}\\
	 \hyperref[cov:45]{\Cov(W_4, W_5)} &= -N p^2 \beta_{T^*} P_B P_A \mu_{T,B} (\mu_{T,A}-\omega_{T,A})
 \end{align*}  
\endgroup
	 
\subsubsection{Variance of $W_1$}\label{var:w1}
\begin{align*}
	\Var(W_1)
	& = \underbrace{\Var(N_{rec}^{PT})}_{(I)} + \underbrace{\Var\left(\hat{\beta}_{T^*}\sum_{i:D_i=1} (1-B_i)\right)}_{(II)} - 2\underbrace{\Cov\left(N_{rec}^{PT}, \hat{\beta}_{T^*} \sum_{i:D_i=1} (1-B_i)\right)}_{(III)}.
\end{align*}
The variance of a nested binomial can be used to compute (I) and one part of (II).
\begin{align*}
	(I) &= (P^*_R)^2 Np(1-p) + Np P_{rec}^{PT} (1-P_{rec}^{PT}) = N p \cdot P_{rec}^{PT} (P_{rec}^{PT}(1-p) + (1-P_{rec}^{PT})) = N pP_{rec}^{PT} (1-pP_{rec}^{PT}) \\
	\Var\left(N_{pos}-N_B \right) &= (1-P_B)^2 Np(1-p) + Np P_B (1 - P_B) = Np(1-P_B)\{1-p(1-P_B)\}.
\end{align*}
Then we can compute the variance of the product in term (II) 
\begin{align*}
	(II) &= \Var(\hat{\beta}_{T^*}) \Var\left(N_{pos}-N_B \right) + \Var(\hat{\beta}_{T^*}) \E(N_{pos}-N_B)^2 + \E(\hat{\beta}_{T^*})^2 \Var(N_{pos}-N_B) \\
	&= \Big\{ \sigma^2_{\hat{\beta}_{T^*}} + \beta^2_{T^*}\Big\}\Big\{(1-P_B)^2 Np(1-p) + Np P_B (1 - P_B) \Big\} + \sigma^2_{\hat{\beta}_{T^*}}  (Np (1-P_B))^2 \\
	&= Np (1-P_B) \Big\{\Big(\sigma^2_{\hat{\beta}_{T^*}} + \beta^2_{T^*} \Big) \Big((1-P_B) (1-p) + P_B \Big) + \sigma^2_{\hat{\beta}_{T^*}} N p (1-P_B)\Big\}.
	\end{align*}
For the covariance term, we use the law of total covariance, and rewrite the conditional covariance term in terms of the recency indicators.
\begin{align*}
	\Cov(N_{rec}^{PT}, \hat{\beta}_{T^*} (N_{pos}-N_B)) &= \beta_{T^*}\Cov\left(\sum_{i:D_i=1} R_i^{PT}, \sum_{i:D_i=1} (1-B_i) \right)\\
	&= \beta_{T^*} N p \left(\E(R_i^{PT}(1-B_i)) - p \E(R_i^{PT}) \E(1-B_i) \right)\\
	&= \beta_{T^*} N p \left(P_{rec}^{PT} - \E(R_i^{PT}B_i) - p P_{rec}^{PT} (1-P_B)\right).
\end{align*}
Recalling the expression for $R_i^{PT}$ in \eqref{mi}, we derive the product of $R_i^{PT}$ and $B_i$:
\begin{align*}
	R_i^{PT}B_i &= \big\{R_i (1-(1-C_i)Q_i \Delta_i) + (1-R_i)C_iQ_i(1-\Delta_i)\big\} (1-C_i)Q_i \\
	&= R_i(1-C_i)Q_i - R_i(1-C_i)Q_i \Delta_i \\
	&= R_i(1-C_i)(1-\Delta_i)Q_i.
\end{align*}
In order to calculate $\E(R_i^{PT}B_i)$, we have a similar derivation to term (III) of the incidence estimator:
\begin{align*}
	\E(R_i^{PT} B_i) &= \E(R_i(1-C_i)Q_i(1-\Delta_i)) \\
	&= \E \left(\int_0^{\infty} \Pr(R_i=1,T_i \geq T^*, \Delta_i=0, Q_i=1|U_i=u) P_i(s,u) du \right) \\
	&= \E \left(Q_i1(T_i \geq T^*)  \frac{\lambda(1-p)}{p} \left\{\int_{0}^{T^*} \Pr(R_i=1|U_i=u) du + \int_{T^*}^{T_i} \Pr(R_i=1|U_i=u) du\right\} \right) \\
	&= \frac{\lambda(1-p)}{p} \E \left(Q_i1(T_i \geq T^*)  \left\{\Omega_{T^*} + \beta_{T^*} (T_i - T^*)\right\} \right) \\
	&= \frac{\lambda(1-p)}{p} \left\{\E (Q_i1(T_i \geq T^*))(\Omega_{T^*} - \beta_{T^*} T^*) + \E(Q_i1(T_i > T^*) T_i) \beta_{T^*} \right\} \\
	&= \frac{\lambda(1-p)}{p} P_B\left\{ (\Omega_{T^*} - \beta_{T^*} T^*) + \mu_{T,B}\beta_{T^*} \right\}.
\end{align*}
Using the relationship that $P_{rec} = \big\{\beta_{T^*} + (\Omega_{T^*} - \beta_{T^*} T^*) \lambda(1-p)/p \big\}$, where $P_{rec}$ is the probability of testing recent on the recency assay, we can rewrite $\lambda(1-p)/p = (P_{rec} - \beta_{T^*})/(\Omega_{T^*} - \beta_{T^*} T^*)$.
Therefore,
\begin{align}
	(III) &= \beta_{T^*} N p \left(P_{rec}^{PT} - \frac{\lambda(1-p)}{p} P_B\left\{ (\Omega_{T^*} - \beta_{T^*} T^*) + \mu_{T,B}\beta_{T^*} \right\} - p P_{rec}^{PT} (1-P_B)\right) \\
	&= \beta_{T^*} N p \left(P_{rec}^{PT} - (P_{rec} - \beta_{T^*}) P_B\left\{ 1 + \mu_{T,B}\beta_{T^*}/(\Omega_{T^*} - \beta_{T^*} T^*) \right\} - p P_{rec}^{PT} (1-P_B)\right)
\end{align}
Putting each of the pieces (I), (II), and (III) together, we have
\begin{align*}
	\Var(W_1) &= N p \Big\{(P_{rec}^{PT})^2(1-p) + P_{rec}^{PT}(1-P_{rec}^{PT}) \\
	&\quad\quad\quad\quad +(1-P_B) \Big\{\Big(\sigma^2_{\hat{\beta}_{T^*}} + \beta^2_{T^*} \Big) \Big((1-P_B) (1-p) + P_B \Big) + \sigma^2_{\hat{\beta}_{T^*}} N p (1-P_B)\Big\} \\
	&\quad\quad\quad\quad -2\beta_{T^*} \cdot \Big\{P_B\left(P_{rec}^{PT} - \big\{\Omega_{T^*} + \beta_{T^*} (\mu_{T,B} - T^*)\big\}  \lambda (1-p)/p\right) + P_{rec}^{PT}(1-P_B) (1-p) \Big\}\Big\} \\
	&= N p \Big\{P_{rec}^{PT}(1-P_{rec}^{PT}) + (1-p) (P_{rec}^{PT} - (1-P_B) \beta_{T^*})^2 \\
	&\quad\quad\quad + \sigma^2_{\hat{\beta}_{T^*}} (1-P_B)\Big\{1-(1-P_B)p + N p (1-P_B) \Big\} \\
	&\quad\quad\quad + \beta_{T^*}P_B\big\{\beta_{T^*}(1-P_B) - 2 \left(P_{rec}^{PT} - (P_{rec} - \beta_{T^*}) \left\{ 1 + \mu_{T,B}\beta_{T^*}/(\Omega_{T^*} - \beta_{T^*} T^*) \right\}\right) \big\} \Big\} .
\end{align*}
This is very similar to the variance of this term in \cite{gao2020sample}, however, there is an extra term at the end led by $\beta_{T^*} P_B$ that comes from the additional variability in $N_{B}$, and the covariance between $N_{B}$ and $N_{R}^*$. Fixing everything else, and only varying the time of non-recent prior tests $B_i$, the variance of $W_1$ increases with $\delta_{T,B}$, i.e., as we use \textit{older} prior tests. In other words, our results are less precise when the prior test results are from a long time ago.\newpage
 
 \newpage
 \subsubsection{Variance of $W_4$}\label{var:w4}
 \begin{align*}
 	\Var(W_4) &= \underbrace{\Var\left(\sum_{i:D_i=1} A_i T_i \right)}_{(I)} + \underbrace{\Var\left(\sum_{i:D_i=1} A_i \int_0^{T_i} \hat{\phi}(u) du\right)}_{(II)} - 2\underbrace{\Cov\left(\sum_{i:D_i=1} A_i T_i, \sum_{i:D_i=1} A_i \int_0^{T_i} \hat{\phi}(u) du\right)}_{(III)}.
 \end{align*}
 For term $(I)$, we have 
 \begin{align*}
 	\Var\left(\sum_{i:D_i=1} A_i T_i \right) &= N p \Big\{\Var(A_i T_i) + (1-p) \E(A_iT_i)^2 \Big\} \\
 	&= N p \Big\{\Var(A_i\E(T_i|A_i)) + \E(A_i\Var(T_i|A_i)) + (1-p) P_A^2 \mu_{T,A}^2 \Big\} \\
 	&= N p \Big\{P_A(1-P_A) \mu_{T,A}^2 + P_A \sigma^2_{T,A} + (1-p) P_A^2 \mu_{T,A}^2 \Big\} \\
 	&= N p \Big\{P_A (\mu_{T,A}^2 + \sigma_{T,A}^2) - pP_A^2 \mu_{T,A}^2 \Big\} \\
 	&= N p P_A \Big\{\sigma^2_{T,A} + \mu_{T,A}^2(1-pP_A)\Big\}.
 \end{align*}
 
 Define $A_1^{N_{pos}} = \{A_i, i = 1, ..., n: D_i=1\}$ and $T_1^{N_{pos}} = \{T_i, i = 1, ..., n: D_i=1\}$. Then,
 \begin{align*}
	(II) &= \Var\left( \sum_{i:D_i=1} A_i \int_0^{T_i} \hat{\phi}(u) du \right)\\
	&= \underbrace{\E\left( \Var \sum_{i:D_i=1} A_i \int_0^{T_i} \hat{\phi}(u) du \Bigg| N_{pos},A_1^{N_{pos}},T_1^{N_{pos}}\right)}_{(\star)} + \underbrace{\Var \left(\E \sum_{i:D_i=1} A_i \int_0^{T_i} \hat{\phi}(u) du \Bigg| N_{pos},A_1^{N_{pos}},T_1^{N_{pos}}\right)}_{(\star\star)}.
\end{align*}
We focus on each of these terms separately. First, fixing $N_{pos} = n$, and all of the $(a_i, t_i)$, we have
\begin{align*}
	\Var \left(\sum_{i=1}^{n} a_i \int_0^{t_i} \hat{\phi}(u) du\right) &= \sum_{i=1}^{n} \sum_{j=1}^{n} a_i a_j\Cov\Big\{\int_0^{t_i} \hat{\phi}(u) du , \int_0^{t_j} \hat{\phi}(u) du \Big\} \\
	&= \sum_{i=1}^{n} \sum_{j=1}^{n}a_i a_j \Bigg\{\E \Big\{\int_0^{t_i} \int_0^{t_j} \hat{\phi}(u) \hat{\phi}(v) du dv \Big\} - \E\Big\{\int_0^{t_i} \hat{\phi}(u) du \Big\}\E\Big\{\int_0^{t_j} \hat{\phi}(u) du \Big\} \Bigg\}\\
	 	&= \sum_{i=1}^{n} \sum_{j=1}^{n}a_i a_j \Bigg\{\int_0^{t_i} \int_0^{t_j} \E[\hat{\phi}(u) \hat{\phi}(v)] du dv -\int_0^{t_i} \int_0^{t_j} \E[\hat{\phi}(u)] \E[\hat{\phi}(u)] du dv\Bigg\} \\
	 	&= \sum_{i=1}^{n} \sum_{j=1}^{n}a_i a_j \Bigg\{\int_0^{t_i} \int_0^{t_j} \Cov[\hat{\phi}(u),\hat{\phi}(v)]du dv\Bigg\} \\
	 	&= \sum_{i=1}^{n} \sum_{j=1}^{n}a_i a_j\int_0^{t_i} \int_0^{t_j} \rho(u, v) du dv.
\end{align*}
Second, fixing the same quantities, we have
\begin{align*}
	\E \left(\sum_{i=1}^{n} a_i \int_0^{t_i} \hat{\phi}(u) du\right) &= \sum_{i=1}^{n} a_i \int_0^{t_i} \phi(u) du.
\end{align*}
Now we take the expectation and variance over $(N_{pos},A_1^{N_{pos}},T_1^{N_{pos}})$:
\begin{align*}
	(\star) &= \E \left(\sum_{i:D_i=1} \sum_{j=1}^{N_{pos}}  A_i A_j \int_0^{t_i} \int_0^{t_j} \rho(u, v) du dv\right)  \\
	&= \E \left(\sum_{i:D_i=1} A_i r_{T_i,T_i}\right) + 2\E \left(\sum_{i\neq j} A_i A_j r_{T_i, T_j}\right) \\
	&= N p P_A r_{T,A} + 2 P_A^2r'_{T,A} \E {N_{pos}\choose 2}\\
	&= N p P_A r_{T,A} +  P_A^2r'_{T,A} \E \{N_{pos}^2 - N_{pos}\} \\
	&= N p  P_A \big\{r_{T,A}  + P_A r'_{T,A} (Np - p)\big\}.
\end{align*}
For the $(\star\star)$ term, we use the same result as for term $(I)$:
\begin{align*}
	(\star\star) &= \Var \left(\sum_{i:D_i=1} A_i \int_0^{T_i} \phi(u) du\right) \\
	&= N p P_A \Big\{\omega_{T,A}^2 (1-pP_A) + \sigma^2_{\omega_{T,A}} \Big\} 
\end{align*}
Thus, $(II) = (\star) + (\star\star) = N p P_A \Big\{r_{T,A}  + P_A r'_{T,A} (Np - p) + \omega_{T,A}^2 (1-pP_A) + \sigma^2_{\omega_{T,A}}\Big\}$. Finally, for term $(III)$, we first note that since $\hat{\phi}$ only appears in the covariance of the second term, we have the following equality:
\begin{align*}
	\Cov\left(\sum_{i:D_i=1} A_i T_i, \sum_{i:D_i=1} A_i \int_0^{T_i} \hat{\phi}(u) du\right) &= \Cov\left(\sum_{i:D_i=1} A_i T_i, \sum_{i:D_i=1} A_i \int_0^{T_i} \phi(u) du\right) \\
	&= Np \left(\E\left(A_i T_i \int_0^{T_i} \phi(u) du\right) - p \E(A_iT_i) \E\left(A_i \int_0^{T_i} \phi(u) du \right)\right) \\
	&= Np P_A \left(\omega_{T,A}^* - p P_A \mu_{T,A} \omega_{T,A} \right).
\end{align*}
Putting all three terms together, we have
\begin{align*}
(I) + (II) - 2*(III) &=  N p P_A \Big\{\sigma^2_{T,A} + \mu_{T,A}^2(1-pP_A)\Big\} \\
&\quad\quad + N p P_A \Big\{r_{T,A}  + P_A r'_{T,A} (Np - p) + \omega_{T,A}^2 (1-pP_A) + \sigma^2_{\omega_{T,A}}\Big\} \\
&\quad\quad - 2 Np P_A \big\{\omega_{T,A}^* - p P_A \mu_{T,A} \omega_{T,A} \big\} \\
&= N p P_A \big \{(\sigma^2_{T,A} + \sigma^2_{\omega_{T,A}}) + (\mu^2_{T,A} + \omega^2_{T,A}) - p P_A (\mu_{T,A} - \omega_{T,A})^2 \\
&\quad\quad\quad\quad\quad + r_{T,A} + P_Ar'_{T,A} (Np - p) - 2\omega_{T,A}^* \big\}.
\end{align*}
 
 \newpage
 \subsubsection{Variance of $W_5$}\label{var:w5}
 \begin{align*}
 	\Var(W_5) &= \Var\left(\hat{\beta}_{T^*} \sum_{i:D_i=1}  B_i T_i\right) \\
 	&= \sigma^2_{\beta_{T^*}} \left\{\E \left(\sum_{i:D_i=1} B_i T_i\right)^2 + \Var \left(\sum_{i:D_i=1} B_i T_i \right)\right\} + \beta_{T^*}^2 \Var\left(\sum_{i:D_i=1} B_i T_i\right).
 \end{align*}
For the variance term below, note that we can use the same equality as term (I) in $\Var(W_4)$, just replacing $P_A$ with $P_B$, $\mu_{T,A}$ with $\mu_{T,B}$, and $\sigma^2_{T,A}$ with $\sigma^2_{T,B}$.
\begin{align*}
	\E \left( \sum_{i:D_i=1} B_i T_i\right) &= Np P_B \mu_{T,B}.\\
	\Var\left(\sum_{i:D_i=1} B_i T_i\right) &= N p P_B \Big\{\mu_{T,B}^2 (1-pP_B) + \sigma^2_{{T,B}} \Big\}.
\end{align*}
Using these in the full expression for $\Var(W_5)$:
\begin{align*}
	\Var(W_5) &= N p P_B \Big\{\sigma^2_{\beta_{T^*}} \mu^2_{T,B}P_B N p + \left(\sigma^2_{\beta_{T^*}} + \beta^2_{T^*} \right) \left(\sigma^2_{T,B} + \mu_{T,B}^2(1-P_Bp) \right)\Big\}.
\end{align*}
\newpage
\subsubsection{Covariance of $W_1, W_2$}\label{cov:12}

\begin{align*}
	\Cov(W_1, W_2) &= \Cov(N_{rec}^{PT}, N_{pos}) - \Cov\left(N_{pos}, \hat{\beta}_{T^*}\sum_{i:D_i=1} (1-B_i) \right).
\end{align*}
For the first and second terms, we have
\begin{align*}
	\Cov(N_{rec}^{PT}, N_{pos}) &= \E(N_{rec}^{PT} N_{pos}) - \E(N_{rec}^{PT}) \E(N_{pos}) \\
	&= \E (N_{pos} \E(N_{rec}^{PT}|N_{pos})) - (N p)^2 P_{rec}^{PT} \\
	&= \E(N_{pos}^2 P_{rec}^{PT}) - \E(N_{pos})^2 P_{rec}^{PT} \\
	&= N p (1-p) P_{rec}^{PT} \\
	\Cov\left(N_{pos}, \hat{\beta}_{T^*}\sum_{i:D_i=1} (1-B_i) \right) &= \beta_{T^*} \E \left(N_{pos} \sum_{i:D_i=1}(1-B_i) \right) - \beta_{T^*} \E (N_{pos}) \E\left(\sum_{i:D_i=1} (1-B_i)\right) \\
	&= \beta_{T^*} (1-P_B) \Var(N_{pos}) \\
	&= N p (1-p) (1-P_B) \beta_{T^*}.
\end{align*}
Putting these together, we get:
\begin{align*}
	\Cov(W_1, W_2) = N p (1-p) \big\{P_{rec}^{PT} -(1-P_B) \beta_{T^*} \big\}.
\end{align*}

\subsubsection{Covariance of $W_1, W_3$}\label{cov:13}
\begin{align*}
	\Cov(W_1, W_3) &= \Cov\left(\hat{\Omega}_{T^*} - \hat{\beta}_{T^*} T^*, N_{rec}^{PT} - \hat{\beta}_{T^*} \sum_{i:D_i=1} (1-B_i) \right) \\
	&= \Cov\left(\hat{\beta}_{T^*} T^*, \hat{\beta}_{T^*} \sum_{i:D_i=1} (1-B_i) \right) \\
	&= T^* N p (1-P_B) \Var(\hat{\beta}_{T^*}).
\end{align*}

\newpage
\subsubsection{Covariance of $W_1, W_4$}\label{cov:14}
\begin{align*}
	\Cov(W_1, W_4) &= \underbrace{\Cov\left(\sum_{i:D_i=1} R_i^{PT}, \sum_{i:D_i=1} A_i T_i\right)}_{(I)} - \underbrace{\Cov\left(\sum_{i:D_i=1} R_i^{PT}, \sum_{i:D_i=1} A_i \int_0^{T_i} \hat{\phi}(u) du\right)}_{(II)} \\
	&\quad\quad - \underbrace{\Cov\left(\hat{\beta}_{T^*} \sum_{i:D_i=1}(1-B_i), \sum_{i:D_i=1} A_i T_i \right)}_{(III)} + \underbrace{\Cov\left(\hat{\beta}_{T^*} \sum_{i:D_i=1}(1-B_i), \sum_{i:D_i=1} A_i \int_0^{T_i} \hat{\phi}(u) du\right)}_{(IV)}.
\end{align*}
We first recognize that since $\hat{\phi}$ and $\hat{\beta}$ only appear in one of the two terms in the covariance each, we can replace them with their expectations (using iterated expectation). We have
\begin{align*}
	(I) &= N p \left\{\E\left(R_i^{PT} A_iT_i  \right) - p\E (R_i^{PT}) \E(A_iT_i) \right\} \\
	&= N p \left\{\E(R_i^{PT}A_iT_i) - pP_{rec}^{PT} P_A \mu_{T,A} \right\} \\
	(II) &= Np \left\{\E\left(R_i^{PT} A_i \int_0^{T_i} \phi(u) du\right) - p \E(R_i^{PT}) \E\left(A_i \int_0^{T_i} \phi(u) du \right) \right\}\\
	&= Np \left\{\E\left(R_i^{PT} A_i \int_0^{T_i} \phi(u) du\right) - p P_{rec}^{PT} P_A \omega_{T,A}\right\} \\
	(III) &= N p \beta_{T^*} \left\{\E \left((1-B_i)A_i T_i \right) - p \E(1-B_i) \E(A_iT_i)\right\} \\
	&= N p \beta_{T^*} \left\{P_A \mu_{T,A} - p(1-P_B) P_A \mu_{T,A}\right\} \\
	&= N p \beta_{T^*} P_A \mu_{T,A} \left\{1-p(1-P_B) \right\} \\
	(IV) &= N p \beta_{T^*} \left\{\E \left((1-B_i)A_i \int_0^{T_i} \phi(u) du \right) - p \E(1-B_i) \E\left(A_i \int_0^{T_i} \phi(u) du\right)\right\} \\
	&= N p \beta_{T^*} \left\{P_A \omega_{T,A} - p(1-P_B) P_A \omega_{T,A}\right\} \\
	&= N p \beta_{T^*} P_A \omega_{T,A} \left\{1-p(1-P_B) \right\}.
\end{align*}
Putting these together, we have
\begin{align*}
	\Cov(W_1,W_4) &= Np \underbrace{\left\{\E\left(R_i^{PT} A_iT_i  \right) -\E\left(R_i^{PT} A_i \int_0^{T_i} \phi(u) du\right)  \right\}}_{(\star)} \\
	&\quad\quad + NpP_A \left\{p P_{rec}^{PT} (\omega_{T,A} - \mu_{T,A}) - \beta_{T^*} (1-p(1-P_B)) (\mu_{T,A} - \omega_{T,A}) \right\}.
\end{align*}
To calculate $(\star)$, we need to calculate $\E\left(R_i^{PT} A_iT_i  \right)$ and $\E\left(R_i^{PT} A_i \int_0^{T_i} \phi(u) du\right)$.
First, we can express $R_i^{PT} A_i$ as:
\begin{align*}
	R_i^{PT}A_i &= \big\{R_i (1-(1-C_i)Q_i \Delta_i) + (1-R_i)C_iQ_i(1-\Delta_i)\big\} C_iQ_i \\
	&= R_i C_i Q_i + (1-R_i) C_iQ_i(1-\Delta_i) \\
	&= C_iQ_i(R_i + (1-R_i)(1-\Delta_i)) \\
	&= C_i Q_i R_i + C_i Q_i (1-R_i)(1-\Delta_i).
\end{align*}
Using similar derivations to that $\E(R_i^{PT}B_i)$ in \ref{var:w1}, we have, for any $f(T_i)$,
\begin{align*}
	\E(C_i Q_i R_i f(T_i)) &= \E\left(1(T_i \leq T^*) Q_i f(T_i) \int_0^{\infty} P(R_i=1|U_i=u) P_i(s,u) du \right) \\
	&= \big\{\beta_{T^*} + (\Omega_{T^*} - \beta_{T^*} T^*) \lambda(1-p)/p \big\} \cdot P_A \E(f(T_i)|A_i=1)
\end{align*}
where we have used the expression for $\int_0^{\infty} P(R_i=1|U_i=u) P_i(s,u) du$ that we derived in Appendix \ref{app:estimator}, equation \eqref{eq:prob-recent}. Recall that this is exactly $P_{rec}$, the probability of recent infection using only the original recency test. Therefore, let $P_{rec} = \big\{\beta_{T^*} + (\Omega_{T^*} - \beta_{T^*} T^*) \lambda(1-p)/p \big\}$. Using this, we have
\begin{align*}
	\E(C_i Q_i R_i T_i) &= P_{rec} \cdot P_A \cdot \mu_{T,A} \\
	\E\left(C_i Q_i R_i \int_0^{T_i} \phi(u) du \right) &= P_{rec} \cdot P_A \cdot \omega_{T,A} \\
	\E(C_i Q_i R_i T_i) - \E\left(C_i Q_i R_i \int_0^{T_i} \phi(u) du \right) &= P_A P_{rec} (\mu_{T,A} - \omega_{T,A})
\end{align*}
Next, we derive $\E(C_i Q_i (1-R_i) (1-\Delta_i) f(T_i))$:
\begin{align*}
	\E(C_i Q_i (1-R_i) (1-\Delta_i) f(T_i)) &= \E\left(f(T_i) \int_0^{\infty} P(T_i \leq T^*, Q_i=1,R_i=0,\Delta_i=0|U_i=u) P_i(s,u) du \right) \\
	&= \E \left(f(T_i) A_i \int_0^{T_i} (1-\phi(u)) P_i(s,u) du \right) \\
	&= \lambda (1-p) /p \E \left(f(T_i) A_i \int_0^{T_i} (1-\phi(u)) du \right).
\end{align*}
Next applying this expression to $f(T_i) = T_i$ and $f(T_i) = \int_0^{T_i} \phi(u) du$, we have
\begin{align*}
	\E(C_i Q_i (1-R_i) (1-\Delta_i) T_i) &= \lambda (1-p) /p \E \left(A_i T_i\int_0^{T_i} (1-\phi(u)) du \right) \\
	&= \lambda (1-p) /p\left\{\E(A_i T_i^2) - \E\left(A_i T_i \int_0^{T_i} \phi(u) du\right) \right\} \\
	\E\left(C_i Q_i (1-R_i) (1-\Delta_i) \int_0^{T_i} \phi(u) du \right) &= \lambda(1-p)/p \left\{\E\left(A_i\int_0^{T_i} \phi(u) du \int_0^{T_i} (1-\phi(u)) du \right) \right\} \\
	&= \lambda(1-p)/p \left\{\E\left(A_iT_i \int_0^{T_i} \phi(u) du \right) - \E\left(A_i \left[\int_0^{T_i} \phi(u)du \right]^2\right) \right\} \\
	\E\left(C_i Q_i (1-R_i) (1-\Delta_i) T_i \right) &-\E\left(C_i Q_i (1-R_i) (1-\Delta_i) \int_0^{T_i} \phi(u) du\right) \\
	&= \frac{\lambda(1-p)}{p} \Bigg\{\E(A_i T_i^2) + \E\left(A_i \left[\int_0^{T_i} \phi(u)du \right]^2\right) \\
	&\quad\quad\quad -2\E\left(A_iT_i\int_0^{T_i}\phi(u) du\right) \Bigg\} \\
	&= \frac{\lambda(1-p)}{p} P_A \left\{\mu_{T,A}^2 + \sigma_{T,A}^2 + \omega_{T,A}^2 + \sigma^2_{\omega_{T,A}} - 2\omega_{TA}^* \right\}.
\end{align*}
Therefore, we have the following expression for $(\star)$:
\begin{align*}
	(\star) = P_A P_{rec} (\mu_{T,A} - \omega_{T,A}) + \frac{\lambda(1-p)}{p} P_A \left\{\mu_{T,A}^2 + \sigma_{T,A}^2 + \omega_{T,A}^2 + \sigma^2_{\omega_{T,A}} - 2\omega_{T,A}^* \right\}
\end{align*}
Combining all of the terms, we have
\begin{align*}
	\Cov(W_1, W_4) &= Np \left\{(\star) + P_A p P_{rec}^{PT} (\omega_{T,A} - \mu_{T,A}) + P_A\beta_{T^*} (1-p(1-P_B)) (\mu_{T,A} - \omega_{T,A}) \right\} \\
	&= Np \left\{(\star) + P_A(\beta_{T^*} (1-p(1-P_B))-pP_{rec}^{PT}) (\mu_{T,A} - \omega_{T,A}) \right\} \\
	&= Np P_A \bigg\{(\mu_{T,A} - \omega_{T,A}) \big(P_{rec} - pP_{rec}^{PT} - \beta_{T^*}(1-p+pP_B) \big) \\
	&\quad\quad\quad\quad\quad + \frac{(P_{rec} - \beta_{T^*})}{(\Omega_{T^*} - \beta_{T^*} T^*)} \big(\mu_{T,A}^2 + \sigma_{T,A}^2 + \omega_{T,A}^2 + \sigma^2_{\omega_{T,A}} - 2\omega_{T,A}^*\big) \bigg\}.
\end{align*}
where we have used the relationship that $P_{rec} = \big\{\beta_{T^*} + (\Omega_{T^*} - \beta_{T^*} T^*) \lambda(1-p)/p \big\}$ to remove dependencies on an estimate of $\lambda$.

\newpage
\subsubsection{Covariance of $W_1, W_5$}\label{cov:15}
\begin{align*}
	\Cov(W_1, W_5) &= \Cov\left(N_{rec}^{PT}, \hat{\beta}_{T^*}\sum_{i:D_i=1} B_i T_i\right) - \Cov\left(\hat{\beta}_{T^*} \sum_{i:D_i=1}(1-B_i), \hat{\beta}_{T^*} \sum_{i:D_i=1} B_i T_i\right) \\
	&= \beta_{T^*} \Cov\left(\sum_{i:D_i=1} R_i^{PT}, \sum_{i:D_i=1} B_i T_i \right) - \left(\sigma^2_{\hat{\beta}_{T^*}} + \beta_{T^*}^2\right)\Cov\left(\sum_{i:D_i=1} (1-B_i), \sum_{i:D_i=1} B_i T_i \right) \\
	&\quad\quad\quad - \sigma^2_{\hat{\beta}_{T^*}} N p \E(1-B_i) N p \E(B_i T_i) \\
	&= N p \Big\{\beta_{T^*} \Big(\E(R_i^{PT}B_iT_i) -p \E(R_i^{PT}) \E(B_iT_i) \Big) \\
	&\quad\quad\quad +\left(\sigma^2_{\hat{\beta}_{T^*}} + \beta_{T^*}^2\right)  p\E(B_iT_i)\E(1-B_i) \\
	&\quad\quad\quad - \sigma^2_{\hat{\beta}_{T^*}} N p\E(1-B_i) \E(B_iT_i)\Big\}.
\end{align*}
We need to calculate $\E(R_i^{PT} B_iT_i)$. Using the same technique as in the derivation of $\E(R_i^{PT} B_i)$ in Section \ref{var:w1}, and replacing $\lambda (1-p) / p$ again, we have
\begin{align*}
	\E(R_i^{PT} B_i T_i) &= \frac{\lambda(1-p)}{p} \left\{\E (Q_i1(T_i > T^*)T_i)(\Omega_{T^*} - \beta_{T^*} T^*) + \E(Q_i1(T_i > T^*) T_i^2) \beta_{T^*} \right\} \\
	&= \frac{(P_{rec} - \beta_{T^*})}{(\Omega_{T^*} - \beta_{T^*} T^*)} P_B\left\{\mu_{T,B}(\Omega_{T^*} - \beta_{T^*} T^*) + (\sigma_{T,B}^2 + \mu^2_{T,B}) \beta_{T^*} \right\}.
\end{align*}
\begin{align*}
	\Cov(W_1, W_5) &= N p P_B \Big\{\beta_{T^*} \Big((P_{rec} - \beta_{T^*}) \Big(\mu_{T,B} + (\Omega_{T^*} - \beta_{T^*} T^*)(\sigma_{T,B}^2 + \mu^2_{T,B}) \beta_{T^*}\Big) - pP_{rec}^{PT} \mu_{T,B} \Big) \\
	&\quad\quad\quad\quad + p \Big(\sigma^2_{\hat{\beta}_{T^*}} + \beta_{T^*}^2 \Big) \mu_{T,B} (1-P_B) - \sigma^2_{\hat{\beta}_{T^*}} N p (1-P_B) \mu_{T,B}\Big\}.
\end{align*}

\newpage
\subsubsection{Covariance of $W_2, W_4$}\label{cov:24}

\begin{align*}
	 		\Cov(W_2, W_4) &= \E\left(N_{pos} \sum_{i:D_i=1} A_i \left( T_i - \int_0^{T_i} \hat{\phi}(u) du\right) \right) - \E\left(N_{pos} \right) \E\left(\sum_{i:D_i=1} A_i\left(T_i - \int_0^{T_i} \hat{\phi}(u) \right)\right) \\
	 		&= (\mu_{T,A} - \omega_{T,A}) P_A \left(\E\left(N_{pos}^2\right) - \E\left(N_{pos}\right)^2\right) \\
	 		&= N p(1-p) P_A (\mu_{T,A} - \omega_{T,A}).
	 	\end{align*}
\subsubsection{Covariance of $W_2, W_5$}\label{cov:25}
\begin{align*}
	\Cov(W_2, W_5) &= \E \left(N_{pos} \hat{\beta}_{T^*} \sum_{i:D_i=1} B_i T_i\right) - \E(N_{pos}) \E \left(\hat{\beta}_{T^*} \sum_{i:D_i=1} B_i T_i \right) \\
	&= \beta_{T^*} \E(B_i T_i) \E(N_{pos}^2) - \E(N_{pos})^2 \beta_{T^*} \E(B_i T_i) \\
	&= N p (1-p) \beta_{T^*} P_B \mu_{T,B}.
\end{align*}

\newpage
\subsubsection{Covariance of $W_3, W_4$}\label{cov:34}
Since $\hat{\beta}_{T^*}$ is independent of everything else, we only need to consider the covariance with $\hat{\Omega}_{T^*}$.
\begin{align*}
	 		\Cov(W_3, W_4) &= \Cov\left(\hat{\Omega}_{T^*}, \sum_{i:D_i=1} A_i \left(T_i -\int_0^{T_i} \hat{\phi}(u) du\right) \right)\\
	 		&= \E \left(\hat{\Omega}_{T^*} \sum_{i:D_i=1} A_i \left( T_i - \int_0^{T_i} \hat{\phi}(u) du \right) \right) - \E\left(\hat{\Omega}_{T^*} \right) \E\left(\sum_{i:D_i=1} A_i \left(T_i - \int_0^{T_i} \hat{\phi}(u) du \right)\right) \\
	 		&= \E\left(\hat{\Omega}_{T^*} \right) \E\left(\sum_{i:D_i=1} A_i \int_0^{T_i} \hat{\phi}(u) du \right) -\E \left(\hat{\Omega}_{T^*} \sum_{i:D_i=1} A_i \int_0^{T_i} \hat{\phi}(u) du \right) \\
	 		&= -\E \left(\sum_{i:D_i=1} A_i \int_0^{T^*} \int_0^{T_i} \E[\hat{\phi}(u) \hat{\phi}(v)] - \E[\hat{\phi}(u)] \E[\hat{\phi}(v)] du dv\right) \\
	 		&= -\E \left(\sum_{i:D_i=1} A_i \int_0^{T^*} \int_0^{T_i} \rho(u, v) du dv \right) \\
	 		&= -N p \cdot P_A \cdot r_{T,A}^*.
	 \end{align*}

\subsubsection{Covariance of $W_3, W_5$}\label{cov:35}
\begin{align*}
	\Cov(W_3, W_5) &= \Cov\left(\hat{\Omega}_{T^*} - \hat{\beta}_{T^*} T^*, \hat{\beta}_{T^*}\sum_{i:D_i=1} B_i T_i \right) \\
	&= -T^*\Cov\left(\hat{\beta}_{T^*}, \hat{\beta}_{T^*} \sum_{i:D_i=1} B_i T_i\right) \\
	&= -T^* \E\left(\sum_{i:D_i=1} B_i T_i \right)\Var(\hat{\beta}_{T^*}) \\
	&= - N p \cdot T^* \cdot P_B \mu_{T,B} \sigma^2_{\hat{\beta}_{T^*}}.
\end{align*}

\newpage
\subsubsection{Covariance of $W_4, W_5$}\label{cov:45}
Recalling that $\hat\beta_{T^*} \perp \hat{\phi} \perp (A_i, B_i, T_i)$, we have that
	\begin{align*}
		\Cov(W_4, W_5) &= \Cov\left(\sum_{i:D_i=1} A_i \left(T_i - \int_0^{T_i} \hat{\phi}(u) du\right), \sum_{i:D_i=1} \hat{\beta}_{T^*} B_i T_i \right)\\
		&= \beta_{T^*} \Cov\left(\sum_{i:D_i=1} A_i \left(T_i - \int_0^{T_i} \phi(u) du\right), \sum_{i:D_i=1} B_i T_i\right) \\
		&= N p \beta_{T^*} \bigg\{\E\left(A_iB_i T_i^2 \right) - p \E(B_i T_i) \E \left(A_i T_i  \right) - \\
		&\quad\quad\quad \E\left(A_iB_i T_i \int_0^{T_i} \phi(u) du \right) + p \E(B_i T_i) \E \left(A_i \int_0^{T_i} \phi(u) du  \right) \bigg\} \\
		&= N p^2 \beta_{T^*} P_B P_A \mu_{T,B} (\omega_{T,A} - \mu_{T,A})
	\end{align*}
	where we have used the fact that $A_i B_i = 0$.
		
\newpage
	 \subsection{Calculating the Variance using the Delta Method}
	 	 With our simplifying assumptions, we have that
	 \begin{align*}
		\log \hat{\lambda} = \log{W_1} - \log(N-W_2) - \log (W_3 + (W_4+W_5)/W_2)
	\end{align*}
	Let $\mathbb{D} = (\E W_4 + \E W_5)/\E W_2$. Then, with $g(W) = \log \hat{\lambda}$, we have:
	\begin{align*}
		\nabla g(W)\Big|_{W=\E W} &= \begin{pmatrix}
			\frac{1}{\E W_1} & \frac{1}{N-\E W_2} + \frac{ \mathbb{D}}{\E W_3 \E W_2 + \E W_4 + \E W_5} & -\frac{1}{\E W_3 + \mathbb{D}} & -\frac{1}{\E W_2(\E W_3 +\mathbb{D})} & -\frac{1}{\E W_2(\E W_3 +\mathbb{D})}
		\end{pmatrix}.
	\end{align*}	
	Thus, using the delta method, the variance of $\log \hat{\lambda}$ is given by the following expression:
	\begin{align*}
		\Var(\log \hat{\lambda}) = \Big\{\nabla g(W)\Big|_{W=\E W}\Big\}^T \Sigma_{W}\Big\{\nabla g(W)\Big|_{W=\E W}\Big\} 
	\end{align*}
	where $\Sigma_{W}$ is the covariance matrix of the components of $W$. We can further find the variance of $\hat{\lambda}$ by applying the delta method again with the transformation $h(\lambda) = \exp(\lambda)$, so the final variance of $\hat{\lambda}$ is $\Var(\hat{\lambda}) = \hat{\lambda}^2 \Var(\log \hat{\lambda})$.

\newpage
\subsection{Comparison with Standard Estimator}

We can compare the variance components with those of the standard estimator by setting $P_B = 0$ and $P_A = 0$, and see that they are identical (see Appendix A of \cite{gao2020sample}).
\begin{align*}
	 \E(W_1) &= N p \cdot \big(P_{rec} - \beta_{T^*} \big) \\
	 \E(W_2) &= Np\\
	 \E(W_3) &= \Omega_{T^*} - \beta_{T^*} T^* \\
	 \E(W_4) &= 0\\
	 \E(W_5) &= 0 \\
	 \hyperref[var:w1]{\Var(W_1)} &= N p \Big\{P_{rec}(1-P_{rec}) + (1-p) (P_{rec} - \beta_{T^*})^2 + \sigma^2_{\hat{\beta}_{T^*}} (1-p + N p ) \Big\} \\	 \Var(W_2) &= Np(1-p)\\
	 \Var(W_3) &= \sigma^2_{\hat{\Omega}_{T^*}} + \sigma^2_{\hat{\beta}_{T^*}} T^{*2}\\
	 \hyperref[var:w4]{\Var(W_4)} &= 0 \\
	 \hyperref[var:w5]{\Var(W_5)} &= 0 \\
	 \hyperref[cov:12]{\Cov(W_1, W_2)} &= N p (1-p) \cdot \big\{P_{rec} - \beta_{T^*} \big\} \\
	 \hyperref[cov:13]{\Cov(W_1, W_3)} &= N p \cdot T^* \sigma^2_{\hat{\beta}_{T^*}} \\
	 \hyperref[cov:14]{\Cov(W_1, W_4)} &= 0 \\
	 \hyperref[cov:15]{\Cov(W_1, W_5)} &= 0 \\
	 \Cov(W_2, W_3) &= 0 \\
	 \hyperref[cov:24]{\Cov(W_2, W_4)} &= 0  \\
	 \hyperref[cov:25]{\Cov(W_2, W_5)} &= 0 \\
	 \hyperref[cov:34]{\Cov(W_3, W_4)} &= 0 \\
	 \hyperref[cov:35]{\Cov(W_3, W_5)} &= 0 \\
	 \hyperref[cov:45]{\Cov(W_4, W_5)} &= 0	 
\end{align*}

\newpage
	 \subsection{Variance Reduction when Using All Non-Recent Tests}
\begingroup
    \fontsize{10pt}{12pt}\selectfont
    We show the derivations for complicated terms that are not already included in \cite{gao2020sample}. To compare this with the variance of the standard estimator, since $P_{rec}^{PT} = \frac{\lambda(1-p)}{p} \left( \Omega_{T^*} - \beta_{T^*} (T^* - P_B \mu_{T,B})) + \beta(1-P_B)\right)$, and $P_{rec} = \frac{\lambda(1-p)}{p} \left( \Omega_{T^*} - \beta_{T^*} T^* \right) + \beta_{T^*}$, we can write
    \begin{align*}
        P_{rec}^{PT} = P_{rec} - \beta_{T^*} P_B \left( 1 - \frac{\lambda(1-p)}{p} \mu_{T,B}\right).
    \end{align*}
	 \begin{align*}
	 \E(W_1) &= N p \cdot \big\{P_{rec}^{PT}\big\} \\
	 \E(W_2) &= Np\\
	 \E(W_3) &= \Omega_{T^*} - \beta_{T^*} T^* \\
	 \E(W_4) &= 0\\
	 \E(W_5) &= N p \cdot \beta_{T^*} \cdot \mu_{T,B} \\
	 \hyperref[var:w1]{\Var(W_1)} &= N p \Big\{P_{rec}^{PT}(1-P_{rec}^{PT}) + (1-p) (P_{rec}^{PT})^2 \\
	&\quad\quad\quad - \beta_{T^*}\big\{2 \left(P_{rec}^{PT} - (P_{rec} - \beta_{T^*}) \left\{ 1 + \mu_{T,B}\beta_{T^*}/(\Omega_{T^*} - \beta_{T^*} T^*) \right\}\right) \big\} \Big\} \\	 \Var(W_2) &= Np(1-p)\\
	 \Var(W_3) &= \sigma^2_{\hat{\Omega}_{T^*}} + \sigma^2_{\hat{\beta}_{T^*}} T^{*2}\\
	 \hyperref[var:w4]{\Var(W_4)} &= 0 \\
	 \hyperref[var:w5]{\Var(W_5)} &= N p \Big\{\sigma^2_{\hat{\beta}_{T^*}} \mu^2_{T,B} N p + \left(\sigma^2_{\hat{\beta}_{T^*}} + \beta^2_{T^*} \right) \left(\sigma^2_{T,B} + \mu_{T,B}^2(1-p) \right)\Big\}\\
	 \hyperref[cov:12]{\Cov(W_1, W_2)} &= N p (1-p) \big\{P_{rec}^{PT} \big\} \\
	 \hyperref[cov:13]{\Cov(W_1, W_3)} &= 0 \\
	 \hyperref[cov:14]{\Cov(W_1, W_4)} &= 0 \\
	 \hyperref[cov:15]{\Cov(W_1, W_5)} &= N p\beta_{T^*} \Big[(P_{rec} - \beta_{T^*}) \Big\{\mu_{T,B} + (\Omega_{T^*} - \beta_{T^*} T^*)(\sigma_{T,B}^2 + \mu^2_{T,B}) \beta_{T^*}\Big\} - pP_{rec}^{PT} \mu_{T,B} \Big] \\
	 \hyperref[cov:24]{\Cov(W_2, W_4)} &= 0  \\
	 \hyperref[cov:25]{\Cov(W_2, W_5)} &= N p (1-p) \cdot \beta_{T^*} \cdot \mu_{T,B} \\
	 \hyperref[cov:34]{\Cov(W_3, W_4)} &= 0 \\
	 \hyperref[cov:35]{\Cov(W_3, W_5)} &= - N p \cdot T^* \cdot \mu_{T,B} \sigma^2_{\hat{\beta}_{T^*}}\\
	 \hyperref[cov:45]{\Cov(W_4, W_5)} &= 0
 \end{align*}  
\endgroup

\section{Additional Simulation Results}\label{app:add-sim}

\begin{enumerate}[label=(\alph*)]
		\item (Table \ref{tab:main}): Main simulation results from Figure \ref{fig:main-sim}, adding the estimated standard error and coverage for the standard and enhanced estimators.
        \item (Table \ref{tab:scenario-c}): Simulation results from sensitivity analysis with non-constant incidence.
        \item (Figure \ref{fig:hptn-selfreport}): Version of Figure \ref{fig:hptn-results} including the enhanced estimator with self-reported prior test results.
        \item (Figure \ref{fig:hptn-validation}): Version of Figure \ref{fig:hptn-results} where we have done a random 50-50 split between the data used for cross-sectional estimators and the longitudinal estimator.
\end{enumerate}

\newpage

\begin{table}[htbp!]
\centering
\caption{Simulation results (5,000 simulated datasets) comparing bias, standard error (SE), standard error estimates (SEE), coverage of the standard and enhanced estimators across range of prior test times (uniformly distributed within the range $(a, b)$) available for $q$ fraction of participants. We also include the percentage reduction in mean squared error over the standard estimator by including prior test results. Bias, SE, and SEE are multiplied by 100.}\label{tab:main}
\begin{tabular}{cc|ccccc}
  \hline
\multirow{2}{*}{$T_i \in (a, b)$} & \multirow{2}{*}{$q$} & \multirow{2}{*}{Bias$\times 10^2$} & \multirow{2}{*}{SE$\times 10^2$} & \multirow{2}{*}{SEE$\times 10^2$} & \multirow{2}{*}{Coverage} & \% Reduction \\
	& & & & & & in MSE
 \\
  \hline
  \multicolumn{7}{c}{\textit{Standard Estimator}} \\
  \hline
-- & -- & 0.10 & 1.01 & 0.98 & 94.84 & 0.00 \\ 
   \hline 
(0, 2) & 0.2 & 0.06 & 0.64 & 0.64 & 94.94 & 59.29 \\ 
  (0, 2) & 0.4 & 0.04 & 0.51 & 0.50 & 94.86 & 74.83 \\ 
  (0, 2) & 0.6 & 0.02 & 0.42 & 0.42 & 95.00 & 82.87 \\ 
  (0, 2) & 0.8 & -0.02 & 0.37 & 0.37 & 94.62 & 86.35 \\ 
  (0, 2) & 1 & 0.00 & 0.34 & 0.33 & 94.98 & 88.74 \\ 
   \hline 
(0, 4) & 0.2 & 0.06 & 0.73 & 0.74 & 95.68 & 47.29 \\ 
  (0, 4) & 0.4 & 0.03 & 0.60 & 0.60 & 95.06 & 65.00 \\ 
  (0, 4) & 0.6 & 0.04 & 0.50 & 0.51 & 95.26 & 75.16 \\ 
  (0, 4) & 0.8 & 0.03 & 0.45 & 0.45 & 94.48 & 79.78 \\ 
  (0, 4) & 1 & 0.02 & 0.41 & 0.40 & 94.96 & 83.60 \\ 
   \hline 
(2, 4) & 0.2 & 0.12 & 0.89 & 0.89 & 95.20 & 20.63 \\ 
  (2, 4) & 0.4 & 0.11 & 0.83 & 0.81 & 95.14 & 31.58 \\ 
  (2, 4) & 0.6 & 0.10 & 0.74 & 0.74 & 95.40 & 45.86 \\ 
  (2, 4) & 0.8 & 0.12 & 0.69 & 0.68 & 95.08 & 51.64 \\ 
  (2, 4) & 1 & 0.08 & 0.63 & 0.62 & 94.98 & 60.74 \\ \hline
\end{tabular}
\end{table}

\newpage

\begin{table}[htbp!]
\centering
\caption{Simulation results comparing bias, standard error, MSE, and coverage of the standard and enhanced estimators when assumption of constant incidence (Assumption \ref{as:constant-inc}) is violated. $q^*$ is the average proportion of tests actually used for the estimator across simulations.}\label{tab:scenario-c}
\begin{tabular}{cccc|c|cccc}
  \hline
 $q$ & Incidence & $\max(T_i)$ & Tests & $q^*$ & Bias$\times 10^2$ & SE$\times 10^2$ & MSE$\times 10^4$ & Coverage \\ 
  \hline
 \multicolumn{9}{c}{\textit{Standard Estimator}} \\
  \hline
-- & Constant & -- & -- & -- & 0.09 & 0.98 & 0.01 & 95.40 \\ 
   \hline 
 \multicolumn{9}{c}{\textit{Enhanced Estimator}} \\
  \hline
0.5 & Constant & 4 & all tests & 0.50 & 0.06 & 0.56 & 0.32 & 94.98 \\ 
  0.5 & Constant & 4 & only recent & 0.25 & 0.03 & 0.61 & 0.37 & 94.62 \\ 
  \hline
  0.5 & Constant & 12 & all tests & 0.50 & 0.08 & 0.66 & 0.44 & 94.66 \\ 
  0.5 & Constant & 12 & only recent & 0.08 & 0.06 & 0.80 & 0.65 & 95.14 \\ 
  \hline
 \multicolumn{9}{c}{\textit{Standard Estimator}} \\
   \hline 
-- & Piecewise & -- & -- & -- & 0.11 & 1.00 & 0.01 & 95.30 \\ 
   \hline 
 \multicolumn{9}{c}{\textit{Enhanced Estimator}} \\
  \hline
0.5 & Piecewise & 4 & all tests & 0.50 & 0.04 & 0.56 & 0.31 & 95.14 \\ 
  0.5 & Piecewise & 4 & only recent & 0.25 & 0.02 & 0.60 & 0.37 & 95.26 \\ 
   \hline 
0.5 & Piecewise & 12 & all tests & 0.50 & 0.15 & 0.66 & 0.46 & 94.84 \\ 
  0.5 & Piecewise & 12 & only recent & 0.08 & 0.09 & 0.81 & 0.66 & 95.24 \\ 
   \hline
   
\end{tabular}\end{table}

\newpage

\begin{figure}[htbp]
	\caption{Incidence estimates for the 21 communities in the HPTN 071 cohort at PC24. The red line is the standard estimator. The purple line is the longitudinal incidence estimator (cases / person-time). The green line is the enhanced estimator based on HPTN 071 study-administered prior HIV tests. The blue line is the enhanced estimator based self-reported data. The table below the plot includes: N Pos, the number of individuals identified as HIV positive ($N_{pos}$); N Rec (RITA), the number identified as RITA-recent ($N_{rec}$); N Rec (PT-RITA) (Study), the number identified as PT-RITA-recent based on study-administered HIV tests ($N_{rec}^{PT}$); N Rec (PT-RITA) (Self.), the number identified as PT-RITA-recent based on self-reported HIV tests ($N_{rec}^{PT}$). Confidence intervals created on the log incidence scale and exponentiated. The variance calculation for the enhanced estimator requires there to be either zero people with non-recent test results, or more than one person with non-recent test results. If there is only one person with non-recent test results, certain terms in the variance formula cannot be calculated. In the one community where that occurred, we removed the single non-recent test result. Communities are ordered by longitudinal incidence estimate (lowest incidence on the left).}
	\centering
	\includegraphics[width=\textwidth]{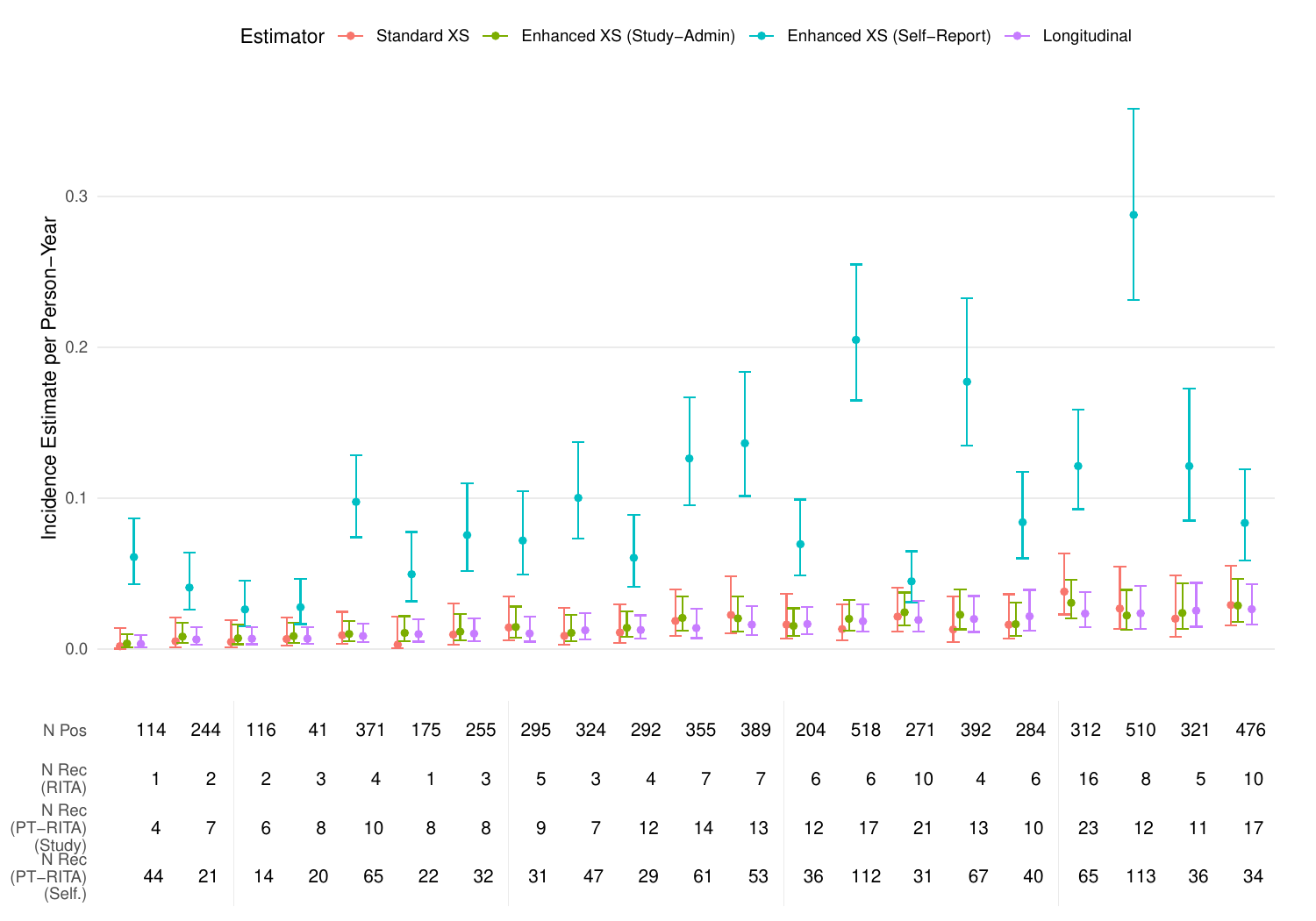}\label{fig:hptn-selfreport}
\end{figure}

\newpage

\begin{figure}[htbp]
	\caption{Incidence estimates for the 21 communities in the HPTN 071 cohort at PC24, split into a random 50-50 training and validation set. The red line is the standard estimator applied to the training set. The green line is the longitudinal incidence estimator (cases / person-time) applied to the validation set. The blue line is the enhanced estimator based on HPTN 071 study-administered prior HIV tests applied to the training set (the dotted line is only using 50\% of available prior test data). Confidence intervals created on the log incidence scale and exponentiated. Communities are ordered by longitudinal incidence estimate (lowest incidence on the left).}
	\centering
	\includegraphics[width=\textwidth]{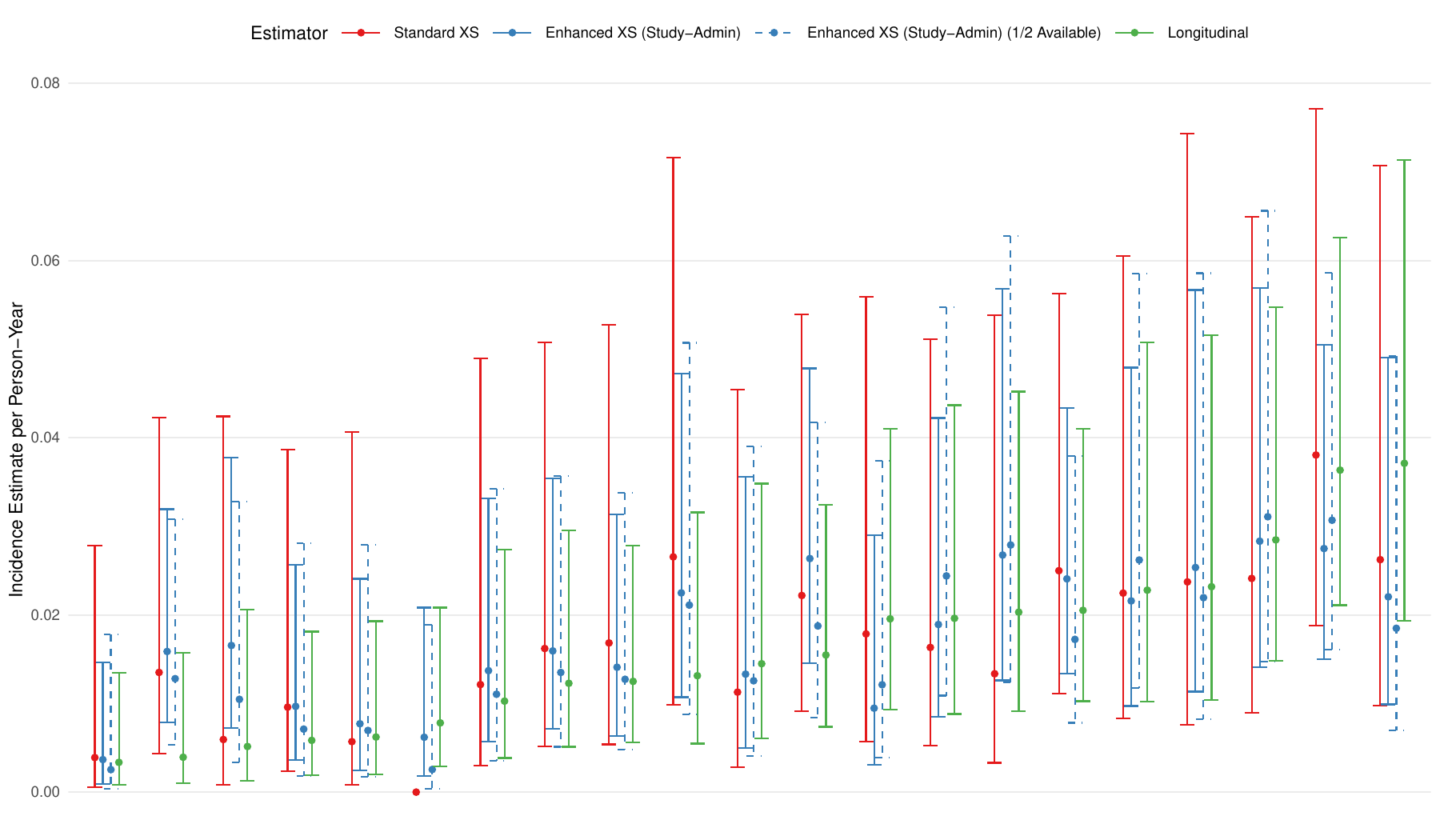}\label{fig:hptn-validation}
\end{figure}
\section{Simulation Derivations}

\subsection{Prior Testing Time and Availability Dependent on Infection Duration}\label{app:sim:duration}

We consider two mechanisms for prior testing. The first mechanism is a baseline level of HIV testing that is uniformly distributed across an interval $[a', b']$, which may be different than the interval for tests that are ultimately used for the enhanced estimator $[a, b]$. Let $T_1 \sim \mathrm{Uniform}(a', b')$ and let $A' \sim \mathrm{Bernoulli}(q')$ indicate receipt of the test $T_1$. We assume that $q'$ proportion of individuals have an HIV test falling into that range.

The second testing mechanism is prompted by a new infection. For a given infection duration $u \in (0, \tau]$, let $E \sim P$ where $P$ is some probability distribution, and $T_2 = u - E$. The ultimate prior testing time is the most recent available, i.e.,
\begin{align*}
	T = \begin{cases}
 	T_1 &\mbox{ if } (T_1 < T_2 \mbox{ or } T_2 < 0) \mbox{ and } A' = 1 \\
 	T_2 &\mbox{ otherwise }
 \end{cases}
\end{align*}
There is a chance that the resulting $T$ will be negative, that is, $T_2$ may fall some time in the future because the test has not been conducted yet and there is not a $T_1$ available. Therefore, we let $Q = 1(T \geq 0)$, the availability of the prior test result. Finally, $\Delta = 1(T \leq u)$. The triplet $(Q, QT, Q\Delta)$ is used as the information available to the enhanced algorithm.

We estimate $P$ by using data from \cite{maheu-girouxDeterminantsTimeHIV2017} on time from HIV infection to linkage to care. Specifically, we use non-linear least squares to fit a Generalized Gamma cumulative distribution function to the time from infection to linkage to care in Figure 1(a) from \cite{maheu-girouxDeterminantsTimeHIV2017}. We include a visual depiction of the testing times based on $T_1$ and $T_2$ in Figure \ref{fig:t1t2-diag}. We include the prior testing time (given that it is available) as a function of $u$, and the availability of the prior testing time as a function of $u$, in Figure \ref{fig:et-u}. 

\begin{figure*}
\caption{Examples of testing times for the two testing mechanisms for someone with a given infection duration $u$. Tests highlighted in green are the tests that would be reported and used in the enhanced algorithm. Tests crossed out in red are not available to use either because they were never conducted ($A' = 0$) or because they would be conducted in the future, past time zero $(T_2 < 0)$. We include the triplet that would be used for the enhanced algorithm at the right side.}\label{fig:t1t2-diag}
	\includegraphics[width=\textwidth]{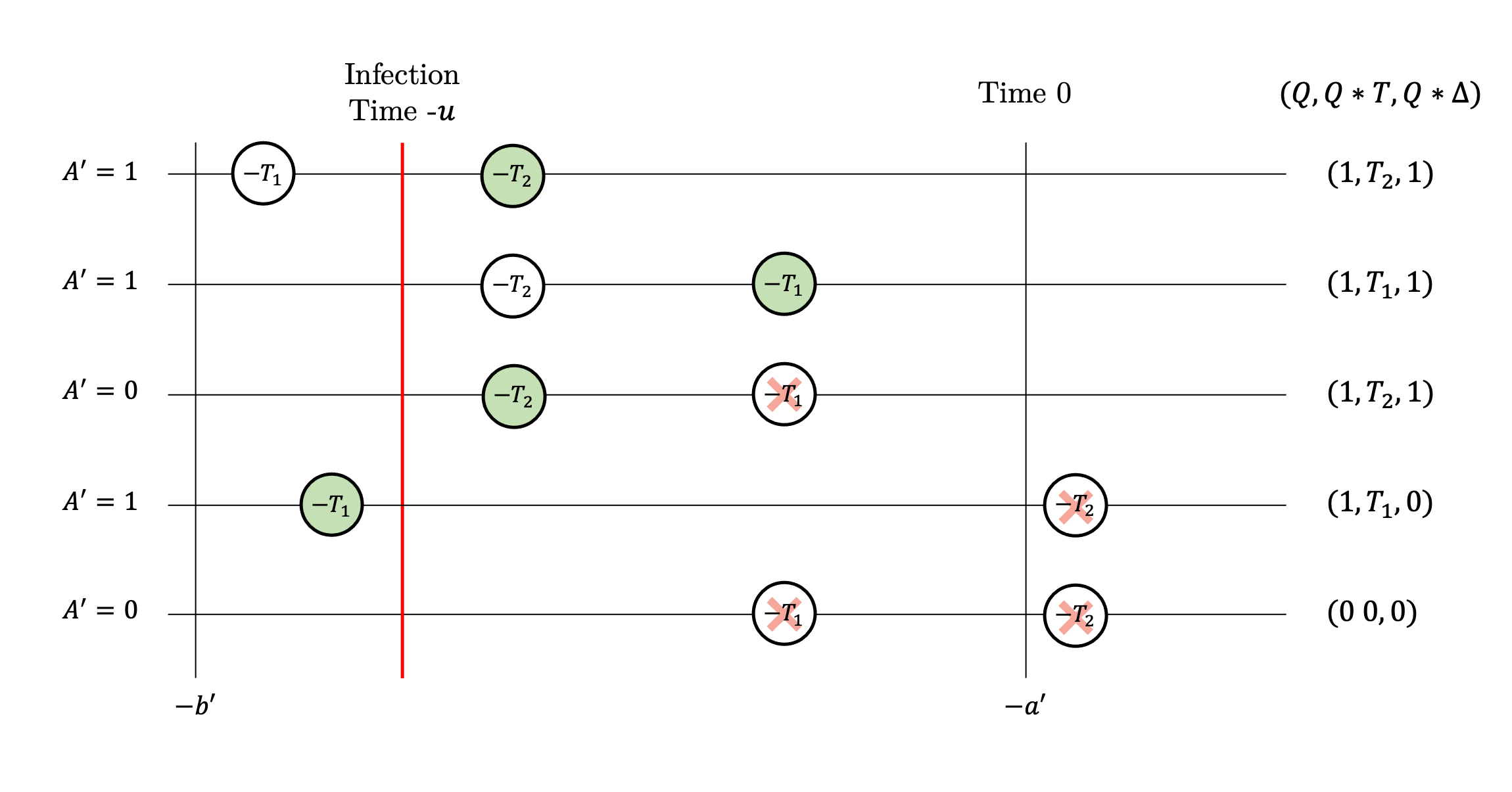}
\end{figure*}

\begin{figure*}[htbp]
\caption{Mean prior testing time for available tests, and availability of prior tests by duration of infection, based on 10,000 simulations. We set $a = 0$, $b = 4$, and $\tau = 10$. The top two plots show the results with baseline HIV testing proportion of $q = 0.4$ and the bottom two show $q = 0.8$.}\label{fig:et-u}
	\includegraphics[width=\textwidth]{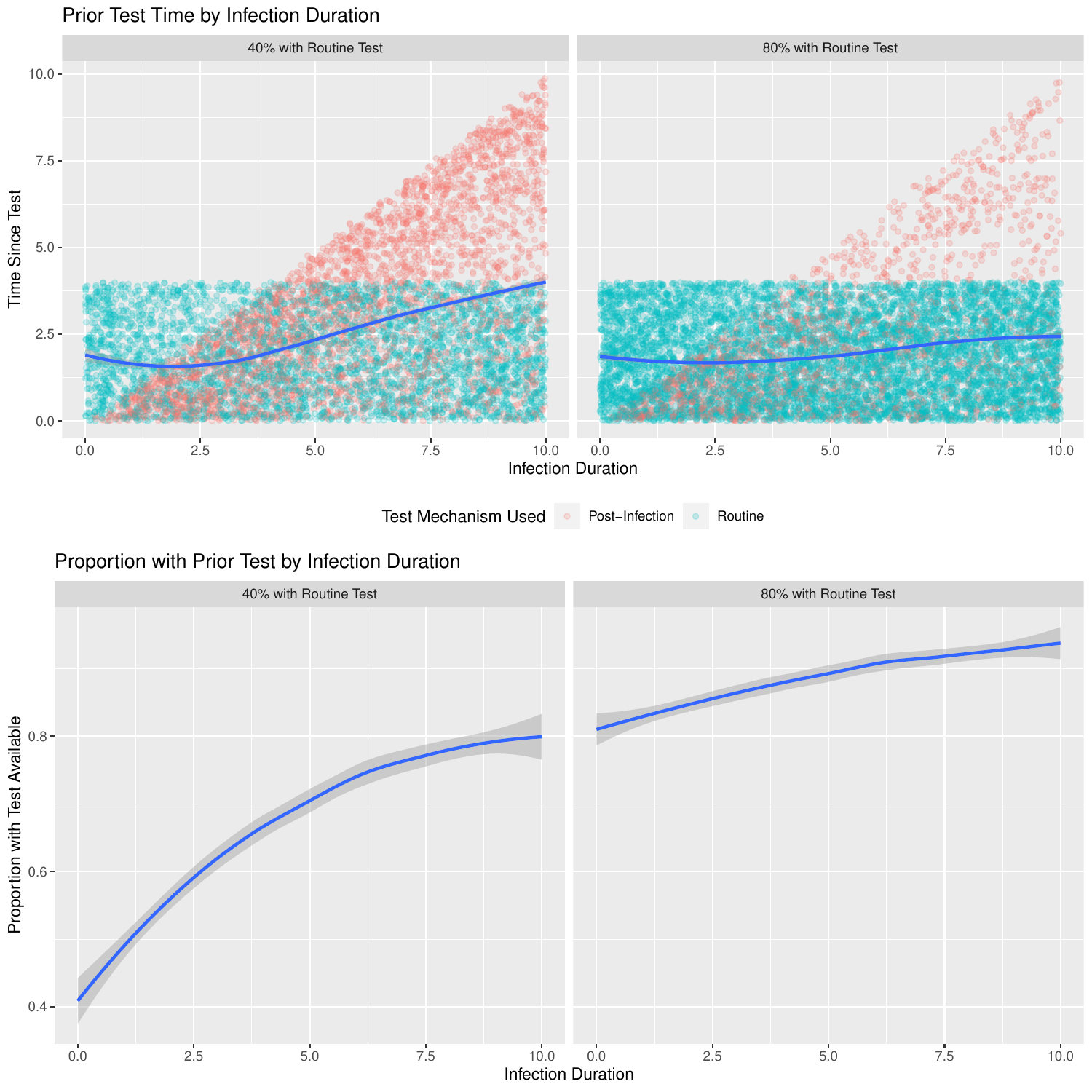}
\end{figure*}

\newpage
\subsection{Derivation of Constant-Linear Incidence Function}\label{app:sim:constant-linear}

We consider the following as our incidence function under misspecification:
\begin{align*}
	\lambda(s) = \lambda(t) + \rho (t - s - T^*) 1(t - s > T^*)
\end{align*}
where $t$ is the current cross-sectional time, and $s \in [t, t-\tau]$ (this is slightly different than $s$ being the cross-sectional time in the main text).
The goal is to simulate historical infection times based on this underlying incidence function. Let $\lambda \equiv \lambda(t)$. Following the setup in Appendix A.2 of \cite{gao2021statistical}, we have that $c_t \geq T^*$ is the solution to the following:
\begin{align}\label{app:sim:ct}
	\frac{1-p}{p} \int_{t-c_t}^{t} \lambda + \rho(t-s-T^*)1(t-s>T^*) ds = 1.
\end{align}
Since the integrand is non-negative everywhere, we can break this up into two parts: first we find the $1-e^*$, which is the value obtained from integrating from $t-T^*$ to $t$. We then use this value to solve for $c_t$. That is, we have
\begin{align*}
	\frac{1-p}{p} \int_{t-T^*}^t \lambda ds &= 1-e^* \\
	\lambda \frac{1-p}{p} T^* &= 1-e^*.
\end{align*}
Then we subtract off $1-e^*$ from the RHS of \eqref{app:sim:ct} to solve for $c_t$.
\begin{align}
\begin{split}\nonumber
	\frac{1-p}{p} \int_{t-c_t}^{t-T^*} \lambda + \rho(t-s-T^*) ds &= e^* \\
	\frac{1-p}{p}\lambda (c_t - T^*) + \rho \frac{1-p}{p} \int_{t-c_t}^{t-T^*} t-s-T^* ds &= 1-\lambda \frac{1-p}{p} T^* \\
	\lambda c_t + \rho \int_{0}^{c_t-T^*} v dv &= \frac{p}{1-p} \quad \text{(letting $v = t-s-T^*$)} \\
	\lambda c_t + \frac{\rho}{2} (c_t - T^*)^2 &= \frac{p}{1-p}
\end{split} \\
\begin{split}\label{app:sim:ctsol}
	\frac{\rho}{2} c_t^2 + c_t (\lambda - \rho T^*) &= \frac{p}{1-p}.
\end{split}
\end{align}
Thus, $c_t$ is the solution to \ref{app:sim:ctsol}. Now we simulate $e \sim \mathrm{Uniform}(0, 1)$. We want to solve for $T$ in the following:
\begin{align*}
	\int_{t-c_t}^{T} \lambda(s) \frac{1-p}{p} ds = e.
\end{align*}
However, if $e > e^*$, then we know that $T$ must be within $[t-T^*, t]$, and it will only have contribution from the constant portion of the incidence function. That is, we have the following:
\begin{align*}
	\int_{t-c_t}^T \lambda(s) \frac{1-p}{p} ds = \underbrace{\int_{t-c_t}^{t-T^*} \lambda(s) \frac{1-p}{p} ds}_{e^*} + \int_{t-T^*}^T \lambda \frac{1-p}{p} ds = e^* + \lambda\frac{1-p}{p}(T-t+T^*) = e.
\end{align*}
Solving the above for $T$, we have
\begin{align*}
	1+\lambda \frac{1-p}{p} (T-t) &= e \implies 
	T = t-\frac{p(1-e)}{\lambda(1-p)}.
\end{align*}
On the other hand, if $e \leq e^*$, then $T \in [t-c_t, t-T^*]$, and we solve the following integral
\begin{align*}
	\int_{t-c_t}^{T} \lambda(s) \frac{1-p}{p} ds = e
\end{align*}
which is equivalent to solving
\begin{align*}
	\int_{T}^{t-T^*} \lambda(s) \frac{1-p}{p} ds &= e^* - e \\
	\lambda (t - T^* - T) + \rho \int_{T}^{t-T^*} t - s - T^* ds&= (e^* - e)\frac{p}{1-p} \\
	\lambda (t-T^* - T) + \rho \int_{0}^{t-T-T^*} v dv&= (e^* - e)\frac{p}{1-p} \\
	\lambda (t - T^* - T) + \frac{\rho}{2} (t - T^* - T)^2 &= (e^* - e) \frac{p}{1-p} \\
	\lambda(t - T) + \frac{\rho}{2} (t - T^* - T)^2 &= \frac{p(1-e)}{1-p} \\
	&= \frac{p(1-e)}{1-p}.
\end{align*}
Expanding $(t-T^*-T)^2$, we have
\begin{align*}
	\frac{\rho}{2}(t-T^*-T)^2 = \rho\big\{t^2/2 - tT - tT^* + TT^* + T^{*2}/2 + T^*/2\big\}.
\end{align*}
Thus, grouping terms, we can solve for $T$ by using the quadratic formula with $a = \frac{\rho}{2}$, $b = (-\lambda - \rho t + \rho T^*)$ and $c = \lambda t + \frac{\rho}{2} t^2  - \rho t T^* + \frac{\rho}{2}T^* - 
\frac{p(1-e)}{(1-p)}$. With these, we have
\begin{align*}
	b^2 = (\lambda^2 + 2\rho t \lambda - 2\lambda \rho T^* - 2\rho^2 t T^* + \rho^2 t^2 + \rho^2 T^{*2}) 
\end{align*}
so the quadratic formula tells us that we have a solution for $T$ of
\begin{align*}
	T &= \lambda/\rho + (t-T^*) - \sqrt{b^2 - 2\rho c}/\rho \\
	&= (t-T^*) - (\sqrt{b^2 - 2\rho c} - \lambda)/\rho.
\end{align*}
Calculating the term inside the square root, we have that
\begin{align*}
	b^2 - 2\rho c &= (\lambda^2 + 2\rho t \lambda - 2\lambda \rho T^* - 2\rho^2 t T^* + \rho^2 t^2 + \rho^2 T^{*2}) - 2\rho \Bigg\{\lambda t + \frac{\rho}{2} t^2  - \rho t T^* + \frac{\rho}{2}T^* - 
\frac{p(1-e)}{(1-p)}\Bigg\} \\
&= (\lambda^2  - 2\lambda \rho T^*) - 2\rho \Bigg\{- 
\frac{p(1-e)}{(1-p)}\Bigg\} \\
&= \lambda^2 + 2\rho \Big\{\frac{p(1-e)}{1-p} - \lambda T^* \Big\}.
\end{align*}
Putting these together, the full expression is
\begin{align*}
	T = (t-T^*) - \Bigg(\sqrt{\lambda^2 + 2\rho \Big\{\frac{p(1-e)}{1-p} - \lambda T^* \Big\}} - \lambda\Bigg)/\rho.
\end{align*}
Therefore, with our simulated $e$, we calculate $T$ using the following:
\begin{align*}
T =
	\begin{cases}
		(t-T^*) - \Bigg(\sqrt{\lambda^2 + 2\rho \Big\{\frac{p(1-e)}{1-p} - \lambda T^* \Big\}} - \lambda\Bigg)/\rho &\mbox{ if } e \leq e^* \\
		t-\frac{p(1-e)}{\lambda(1-p)} &\mbox{ if } e > e^*.
	\end{cases}
\end{align*}
The two are equivalent when $e = e^*$ because $p(1-e^*)/(1-p) = \lambda T^*$, so the first case is $t - T^*$ since the inside of the square root is zero, and the second case is $t - \lambda T^*/\lambda = t-T^*$ as well. Thus, if $e = e^*$, we could simply set $T = t - T^*$. The incidence and infection time distributions corresponding to the above derivations are show in Figure \ref{app:fig:lincon-inc}.

\begin{figure}[htbp]
	\caption{Incidence function and distribution of 10,000 simulated infection times for the constant incidence function $\lambda(s) = \lambda$, and constant-linear incidence function $\lambda(s) = \lambda + \rho (t - s - T^*) 1(t - s > T^*)$, with $\rho = 0.0039$ and $\lambda = 0.032$ based on data fit to \cite{pattanasinRecentDeclinesHIV2020}.}\label{app:fig:lincon-inc}
	\includegraphics[width=\textwidth]{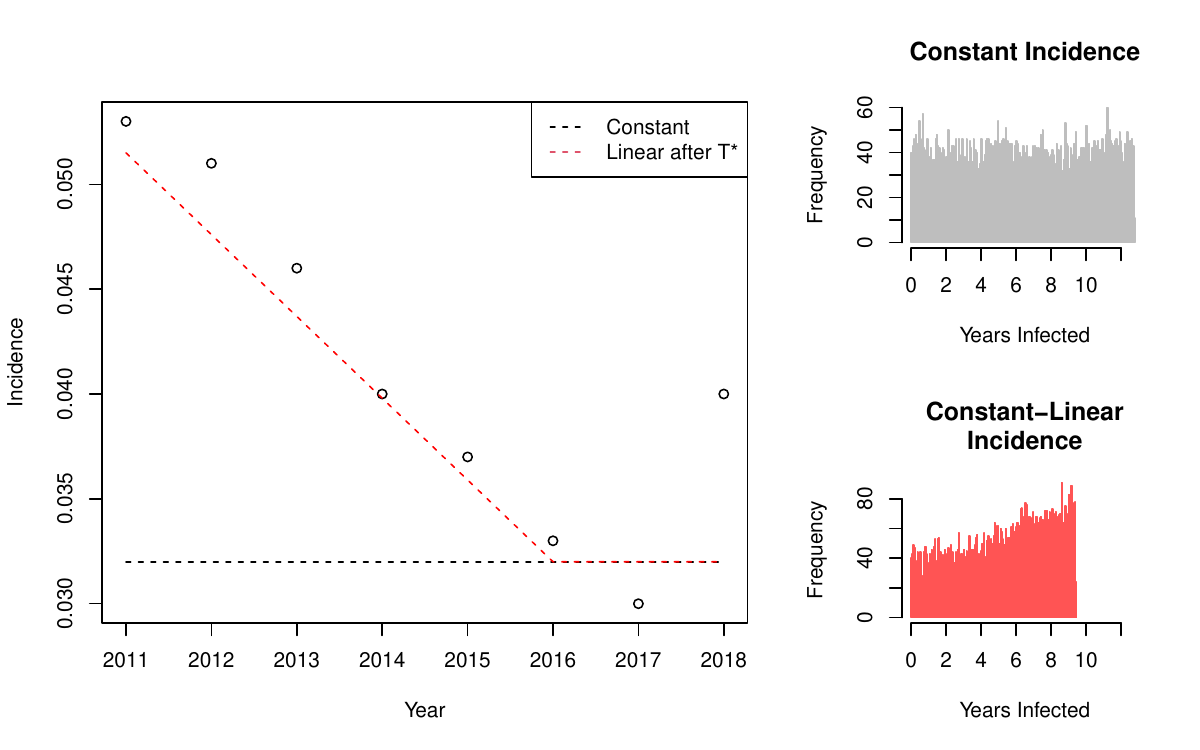}
\end{figure}

\newpage
\printbibliography